\documentclass[twocolumn,journal,twoside]{IEEEtran}
\usepackage[T1]{fontenc}
\usepackage[bookmarks=false,draft]{hyperref}
\usepackage[cmex10]{amsmath}
\usepackage{subcaption}

\usepackage{algorithm}
\usepackage[noend]{algpseudocode}

\usepackage{cite}

\usepackage{graphicx}
\usepackage{pstricks}
\usepackage{amssymb}
\usepackage{amsthm}
\usepackage{mathtools}

\usepackage{relsize}
\usepackage{bbm}
\usepackage{bm}
\usepackage[font={small}]{caption}

\makeatletter
\def\thm@space@setup{\thm@preskip=2pt
	\thm@postskip=2pt \itshape}
\makeatother

\newtheoremstyle{newstyle}      
{} 
{} 
{\mdseries} 
{} 
{\bfseries} 
{.} 
{ } 
{} 

\theoremstyle{newstyle}

\newtheorem{theorem}{Theorem}
\newtheorem{lemma}{Lemma}

\newtheorem{corollary}{Corollary}

\theoremstyle{definition}

\newtheorem{definition}{Definition}

\theoremstyle{remark}
\newtheorem{remark}{Remark}

\IEEEoverridecommandlockouts

\begin{document}
	\sloppy
	
	\setlength{\belowcaptionskip}{-6pt}
	\setlength{\abovedisplayskip}{1mm}
	\setlength{\belowdisplayskip}{1mm}
	\setlength{\abovecaptionskip}{1mm}
	
	\title{The Exact Rate-Memory Tradeoff\\ for Caching with Uncoded Prefetching}
	\author{Qian~Yu,~\IEEEmembership{Student~Member,~IEEE,}
	Mohammad~Ali~Maddah-Ali,~\IEEEmembership{Member,~IEEE,}
        and~A.~Salman~Avestimehr,~\IEEEmembership{Senior Member,~IEEE}
        \thanks{Manuscript received September 26, 2016; revised August 09, 2017; accepted November 29, 2017.
A shorter version of this paper was presented at ISIT, 2017 \cite{yu2016exactisit}. }
\thanks{Q.~Yu and A.S.~Avestimehr are with the Department of Electrical Engineering, University of Southern California, Los Angeles, CA, 90089, USA (e-mail:  qyu880@usc.edu; avestimehr@ee.usc.edu).}
\thanks{M. A. Maddah-Ali is with Department of Electrical Engineering, Sharif University of Technology, Tehran, 11365, Iran (e-mail: maddah\_ali@sharif.edu).}
\thanks{
Communicated by P. Mitran, Associate Editor for Shannon Theory. }
\thanks{
This work is in part supported by NSF grants CCF-1408639, NETS-1419632, and ONR award N000141612189. }
\thanks{Copyright (c) 2017 IEEE. Personal use of this material is permitted.  However, permission to use this material for any other purposes must be obtained from the IEEE by sending a request to pubs-permissions@ieee.org.}}
	\maketitle
	
	\begin{abstract}
	
	We consider a basic cache network, in which a single server is connected to multiple users via a shared bottleneck link. The server has a database of 
	files (content). Each user has an isolated memory that can be used to cache content in a prefetching phase. In a following delivery phase, each user requests a file from the database, and the server needs to deliver users' demands as efficiently as possible  by taking into account their cache contents. 
	We focus on an important and commonly used class of prefetching schemes, where the caches are filled with uncoded data. We provide the exact characterization of the rate-memory tradeoff for this problem, by deriving both the \emph{minimum average rate} (for a uniform file popularity) and the \emph{minimum peak rate} required on the bottleneck link for a given cache size available at each user.
	In particular, we propose a novel caching scheme, which strictly improves the state of the art by exploiting commonality among user demands. We then demonstrate the exact optimality of our proposed scheme through a matching converse, by dividing the set of all demands into types, and showing that the placement phase in the proposed caching scheme is universally optimal for all types. 	Using these techniques, we also fully characterize the rate-memory tradeoff for a decentralized setting, in which users fill out their cache content  without any coordination.

	\end{abstract}
	

	\section{Introduction}
	Caching is a commonly used approach to reduce traffic rate in a network system during peak-traffic times, by duplicating part of the content in the memories distributed across the network. In its basic form, a caching system operates in two phases: (1) a placement phase, where each cache is populated up to its size, and (2) a delivery phase, where the users reveal their requests for content and the server has to deliver the requested content. During the delivery phase, the server exploits the content of the caches to reduce network traffic.
	
	Conventionally, caching systems have been based on uncoded unicast delivery where the objective is mainly to maximize the hit rate, i.e. the chance that the requested content can be delivered locally
	  \cite{sleator1985amortized, dowdy82,almeroth96,dan96,korupolu99,meyerson01,baev08, borst10}. While in systems with single cache memory this approach can achieve optimal performance, it has been recently shown in \cite{maddah-ali12a} that for multi-cache systems, the optimality no longer holds. In \cite{maddah-ali12a}, 
	 an information theoretic framework for 
	 multi-cache systems was introduced, and it was shown that coding can offer a significant gain that scales with the size of the network.
	Several coded caching schemes have been proposed since then \cite{DBLP:journals/corr/Chen14h, wan2016caching, sahraei2016k, tian2016caching, amiri2016fundamental, amiri2016coded}. 
		 The caching problem has also been extended in various directions, including decentralized caching \cite{maddah-ali13}, online caching \cite{pedarsani13}, caching with nonuniform demands \cite{niesen13, zhang15, ji2015order,ramakrishnan2015efficient}, hierarchical caching \cite{hachem14, karamchandani14, hachem15}, device-to-device caching \cite{ji14b}, cache-aided interference channels \cite{maddah2015cache, naderializadeh2016fundamental, hachem2016layered, DBLP:journals/corr/HachemND16a}, caching on file selection networks \cite{wang2015information, DBLP:journals/corr/WangLG15, lim2016information}, caching on broadcast channels \cite{timo2015joint, bidokhti2016erasure, bidokhti2016noisy, bidokhti2016upper}, and caching for channels with delayed feedback with channel state information  \cite{zhang2015fundamental, zhang2015coded}. The same idea is also useful in the context of distributed computing, in order to take advantage of extra computation to reduce the communication load \cite{2016arXiv160407086L, globedcd16, li2016scalable,7901473, yu2017howto}.

		
	Characterizing the exact rate-memory tradeoff in the above caching scenarios is an active line of research. Besides developing better achievability schemes, there have been efforts in tightening the outer bound of the rate-memory tradeoff 
	\cite{ghasemi15,   lim2016information, sengupta15, DBLP:journals/corr/WangLG16, tian2016symmetry, prem2015critical}. Nevertheless, in almost all scenarios, there is still a gap between the state-of-the-art communication load and the converse, leaving the exact rate-memory tradeoff an open problem.

		In this paper, we focus on an important class of caching schemes, where the prefetching scheme is required to be uncoded. In fact, almost all caching schemes proposed for the above mentioned problems use uncoded prefetching.
	 As a major advantage, uncoded prefetching allows us to handle asynchronous demands without increasing the communication rates, by dividing files into smaller subfiles  \cite{maddah-ali13}. 
		Within this class of caching schemes, we characterize the exact rate-memory tradeoff for both the \emph{average rate} for uniform file  popularity  and the \emph{peak rate}, in both centralized and decentralized settings, for all possible parameter values. 

		In particular, we first propose a novel caching strategy for the centralized setting (i.e., where the users can coordinate in designing the caching mechanism, as considered in  \cite{maddah-ali12a}),  which strictly improves the state of the art, reducing both the average rate and the peak rate.
		We exploit commonality among user demands by showing that the scheme in \cite{maddah-ali12a} may introduce redundancy in the delivery phase, and proposing a new scheme that effectively removes all such redundancies in a systematic way.

	 In addition, we demonstrate the exact optimality of the proposed scheme through a matching converse. The main idea is to divide the set of all demands into smaller subsets (referred to as types), and derive tight lower bounds for the minimum peak rate and the minimum average rate on each type separately. We show that, when the prefetching is uncoded, the rate-memory tradeoff can be completely characterized using this technique, and the placement phase in the proposed caching scheme universally achieves those minimum rates on all types. 
	
	Moreover, we extend the techniques we developed for the centralized caching problem to characterize the exact rate-memory tradeoff in the decentralized setting (i.e. where the users cache the contents independently without any coordination, as considered in \cite{maddah-ali13}). 
	Based on the proposed centralized caching scheme, we develop a new decentralized caching scheme that strictly improves the state of the art \cite{maddah-ali13, amiri2016coded}. In addition, we formally define the framework of decentralized caching, and prove matching converses given the framework, showing that the proposed scheme is optimal.

	To summarize, the main contributions of this paper are as follows:
	
	\begin{itemize}
	    \item Characterizing the rate-memory tradeoff for average rate, by developing a novel caching design and proving a matching information theoretic converse.
	    \item Characterizing the rate-memory tradeoff for peak rate, 
	    by extending the achievability and converse proofs to account for the worst case demands.
	    \item Characterizing the rate-memory tradeoff for both average rate and peak rate in a decentralized setting, where the users cache the contents independently without coordination.
	   	\end{itemize}
	Furthermore, in one of our recent works \cite{yu2017characterizing}, we have shown that the achievablity scheme we developed in this paper also leads to the yet known tightest characterization (within factor of $2$) in the general problem with coded prefetching, for both average rate and peak rate, in both centralized and decentralized settings.
	
	 The problem of caching with uncoded prefetching was initiated 
	 in \cite{kai2016optimality,wan2016caching}, which showed that the scheme in \cite{maddah-ali12a} is optimal when considering \emph{peak rate} and \emph{centralized caching}, if there are more files than users. Although not stated in \cite{kai2016optimality,wan2016caching}, the converse bound in our paper for the special case of peak rate and centralized setting could have also been derived using their approach. 	 
	 	 In this paper however, we introduce the novel idea of demand types, which allows us to go beyond and characterize the rate-memory tradeoff for both peak rate and average rate for all possible parameter values, in both centralized and decentralized settings. Our result covers the peak rate centralized setting, as well as strictly improves the bounds in all other cases.
	 More importantly,  
	 we introduce a new achievability scheme, which strictly improves the scheme in~\cite{maddah-ali12a}.

	 
	 The rest of this paper is organized as follows. Section \ref{sec:sys} formally establishes a centralized caching framework, and defines the main problem studied in this paper. Section \ref{sec:main} summarizes the main result of this paper for the centralized setting.
	 Section \ref{sec:opt} describes and demonstrates the optimal centralized caching scheme that achieves the minimum expected rate and the minimum peak rate. Section \ref{sec:conv} proves matching converses that show the optimality of the proposed centralized caching scheme. Section \ref{sec:ext} extends the techniques we developed for the centralized caching problem to characterize the exact rate-memory tradeoff in the decentralized setting.

		\section{System Model and Problem Definition}\label{sec:sys}
	
	In this section, we formally introduce the system model for the centralized caching problem. Then, we define the rate-memory tradeoff based on the introduced framework, and state the main problem studied in this paper.
	
	\subsection{System Model}
	We consider a system with one server connected to $K$ users through a shared, error-free link (see Fig. \ref{fig:net}). The server has access to a database of $N$ files $W_1, . . . , W_N$, each of size $F$ bits.\footnote{Although we only focus on binary files, the same techniques developed in this paper can also be used for cases of q-ary files and files using a mixture of different alphabets, to prove that same rate-memory trade off holds.} We denote the $j$th bit in file $i$ by $B_{i,j}$, and we assume that all bits in the database are i.i.d. Bernoulli random variables with $p=0.5$.
	Each user has an isolated cache memory of size $MF$ bits, where $M\in[0,N]$. For convenience, we define parameter $t=\frac{KM}{N}$.
	
    	\begin{figure}[htbp]
    		\centering
    		\includegraphics[width=0.4\textwidth]{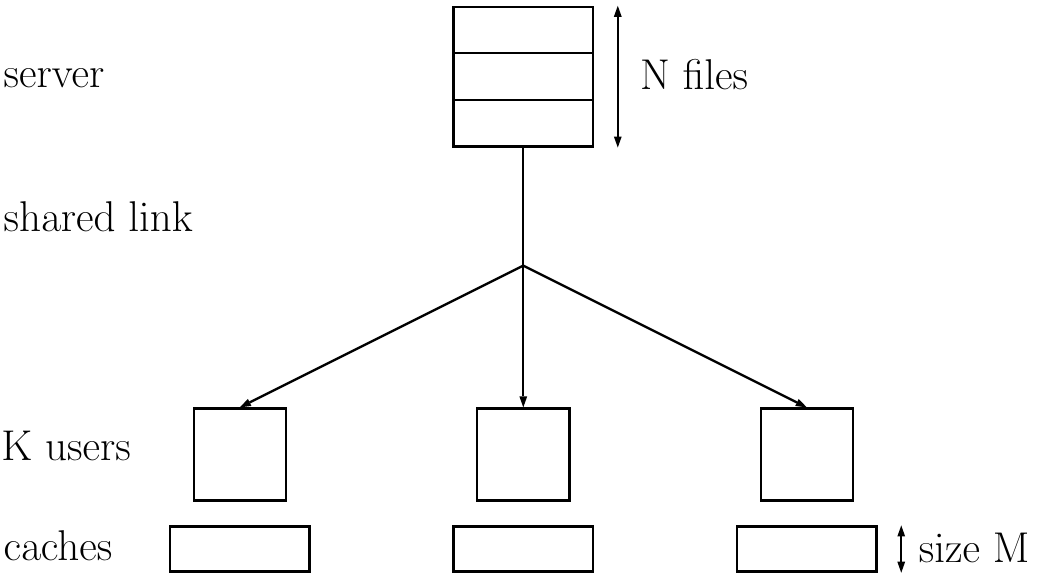}
    		\caption{Caching system considered in this paper. The figure illustrates the case where $K=N=3$, $M=1$.}
    		\label{fig:net}
    	\end{figure}

	The system operates in two phases, a placement phase and a delivery phase. 
	In the placement phase, users are given access to the entire database, and each user can fill their cache using the database. However, instead of allowing coding in prefetching \cite{maddah-ali12a}, we focus on an important class of prefetching schemes, referred to as uncoded prefetching schemes:
	\theoremstyle{definition}
    \begin{definition} An \textit{uncoded prefetching scheme} is where each user $k$ selects no more than $MF$ bits from the database and stores them in its own cache, without coding. Let $\mathcal{M}_k$ denote the set of indices of the bits chosen by user $k$, then we denote the prefetching as $$\boldsymbol{\mathcal{M}}=(\mathcal{M}_1,...,\mathcal{M}_K).$$
    \end{definition}
	In the delivery phase, only the server has access to the
	database. Each user $k$ requests one of the files 
	in the database. To characterize user requests, we define \textit{demand} $\boldsymbol{d}=\left(d_1,...,d_K\right)$, where $d_k$ is the index of the file requested by user $k$. We denote the number of distinct requested files in $\boldsymbol{d}$ by $N_{\textup{e}}(\boldsymbol{d})$, and denote the set of all possible demands by $\mathcal{D}$, i.e.,  $\mathcal{D}=\{1,...,N\}^K$.

	The server is informed of the demand and proceeds by generating a signal $X$ of size $RF$ bits as a function of $W_{1},...,W_{N}$, and transmits the signal over the shared link. $R$ is a fixed real number given the demand $\boldsymbol{d}$. The values $RF$ and $R$ are referred to as the load and the rate of the shared
	link, respectively. Using the values of bits in $\mathcal{M}_k$ and the signal $X$ received over the shared
	link, each user $k$ aims to reconstruct their requested file $W_{d_k}$. 
		
	\subsection{Problem Definition}	
	 
		Based on the above framework, we define the rate-memory tradeoff for the average rate using the following terminology. Given a prefetching $\boldsymbol{\mathcal{M}}=(\mathcal{M}_1,...,\mathcal{M}_K)$, we say a communication rate $R$ is \textit{$\epsilon$-achievable} for demand $\boldsymbol{d}$ if and only if there exists a message $X$ of length $RF$ such that every active user $k$ is able to recover its desired file $W_{d_k}$ with a probability of error of at most $\epsilon$. This is rigorously defined as follows:  
		
		\begin{definition}
		$R$ is \textit{$\epsilon$-achievable} given a prefetching $\boldsymbol{\mathcal{M}}$ and a demand $\boldsymbol{d}$ if and only if we can find an encoding function $\psi: \{0,1\}^{NF}\rightarrow \{0,1\}^{RF}$ that maps the $N$ files to the message:
		\begin{align}
		    X=\psi(W_1,...,W_N), \nonumber
		\end{align}
		and $K$ decoding functions $\mu_k: \{0,1\}^{RF}\times \{0,1\}^{|\mathcal{M}_k|}\rightarrow \{0,1\}^{F}$ that each map the signal $X$ and the cached content of user $k$ to an estimate of the requested file $W_{d_k}$, denoted by $\hat{W}_{\boldsymbol{d},k}$:
		\begin{align}
		    \hat{W}_{\boldsymbol{d},k}=\mu_k(X,\{B_{i,j}\ |\ (i,j)\in \mathcal{M}_k\}), \nonumber
		\end{align}
		such that 
		\begin{align}
		    \mathbbm{P} (\hat{W}_{\boldsymbol{d},k}\neq W_{d_k} )\leq \epsilon. \nonumber
		\end{align}
		\end{definition}

		We denote $R_\epsilon^*(\boldsymbol{d},\boldsymbol{\mathcal{M}})$ as the minimum $\epsilon$-achievable rate given $\boldsymbol{d}$ and $\boldsymbol{\mathcal{M}}$. 
		Assuming that all users are making requests independently, and that all files are equally likely to be requested by each user, the probability distribution of the demand $\boldsymbol{d}$ is uniform on $\mathcal{D}$. We define the average rate $R_\epsilon^*(\boldsymbol{\mathcal{M}})$ as the expected minimum achievable rate given a prefetching $\boldsymbol{\mathcal{M}}$ under uniformly random demand, i.e., 
		\begin{equation}
		R_\epsilon^*(\boldsymbol{\mathcal{M}})=\mathbb{E}_{\boldsymbol{d}}[ R_\epsilon^*(\boldsymbol{d},\boldsymbol{\mathcal{M}})].\nonumber
		\end{equation}
	
		 The rate-memory tradeoff for the average rate is essentially finding the minimum average rate $R^*$, for any given memory constraint $M$, that can be achieved by prefetchings satisfying this constraint with vanishing error probability for sufficiently large file size. Rigorously, we want to find
		\begin{equation}
			R^*=\sup_{\epsilon>0} \adjustlimits\limsup_{F\rightarrow+\infty}\min_{\boldsymbol{\mathcal{M}}} R^*_\epsilon(\boldsymbol{\mathcal{M}}).\nonumber
		\end{equation}
		as a function of $N$, $K$, and $M$.
		
		Similarly, the rate-memory tradeoff for peak rate is essentially finding the minimum peak rate, denoted by $R^*_{\textup{peak}}$, which is formally defined in Appendix \ref{app:worst}.

\section{Main Results}\label{sec:main}

We state the main result of this paper in the following theorem.
\begin{theorem}\label{th:p}
	
	For a caching problem with $K$ users, a database of $N$ files, local cache size of $M$ files at each user, and parameter $t=\frac{KM}{N}$, we have
	\begin{align}\label{eq:unif_pre}
	R^*=\mathbb{E}_{\boldsymbol{d}}\left[ \frac{\binom{K}{t+1}-\binom{K-N_{\textup{e}}(\boldsymbol{d})}{t+1}}{\binom{K}{t}}\right],
	\end{align}
	for $t\in\{0,1,...,K\}$, where $\boldsymbol{d}$ is uniformly random on $\mathcal{D}=\{1,...,N\}^K$ and $N_{\textup{e}}(\boldsymbol{d})$ denotes the number of distinct requests in $\boldsymbol{d}$. Furthermore, for $t\notin\{0,1,...,K\}$, $R^*$ equals the lower convex envelope of its values at $t\in\{0,1,...,K\}$.\footnote{In this paper we define $\binom{n}{k}=0$ when $k>n$.}
\end{theorem}

\begin{remark} 
	To prove Theorem \ref{th:p}, we propose a new caching scheme that strictly improves the state of the art \cite{maddah-ali12a}, which was relied on by all prior works considering the minimum average rate for the caching problem \cite{niesen13, zhang15,ji2015order,lim2016information}. 
	In particular, the rate achieved by the previous best known caching scheme 
	 equals the lower convex envelope of $\min\{\frac{K-t}{t+1}, \mathbb{E}_{\boldsymbol{d}}[N_{\textup{e}}(\boldsymbol{d})(1-\frac{t}{K})] \}$ at $t\in\{0,1,...,K\}$, which is strictly larger than $R^*$ when $N>1$ and $t<K-1$. 
	 	For example, when $K=30$, $N=30$, and $t=1$, the state-of-the-art scheme requires a communication rate of $14.12$, while the proposed scheme achieves the rate $12.67$, both rounded to two decimal places. 
	 	
	 	The improvement of our proposed scheme over the state of the art can be interpreted intuitively as follows. The caching scheme proposed in \cite{maddah-ali12a} essentially decomposes 
	  the problem into $2$ cases: in one case, the redundancy of user demands is ignored, and the information is delivered by satisfying different demands using single coded multicast transmission
	; in the other case, random coding is used to deliver the same request to multiple receivers
	. Our result demonstrates that the decomposition of the caching problem into these 2 cases is suboptimal, and our proposed caching scheme precisely accounts for the effect of redundant user demands.   
\end{remark}

\begin{remark}
The technique for finding the minimum average rate in the centralized setting can be straightforwardly extended to find the minimum peak rate, which was solved for $N\geq K$ \cite{kai2016optimality}. Here we show that we not only recover their result, but also fully characterize the rate for all possible values of $N$ and $K$, resulting in the following corollary, which will be proved in Appendix \ref{app:worst}.
\end{remark} 

\begin{corollary}
\label{cor:worst}
		For a caching problem with $K$ users, a database of $N$ files, a local cache size of $M$ files at each user, and parameter $t=\frac{KM}{N}$, we have 
		\begin{align}\label{eq:cent}
		R^*_{\textup{peak}}=
		\frac{\binom{K}{t+1}-\binom{K-\min\{K,N\}}{t+1}}{\binom{K}{t}}
		\end{align}
		for $t\in\{0,1,...,K\}$. Furthermore, for  $t\notin\{0,1,...,K\}$, $R^*_{\textup{peak}}$ equals the lower convex envelope of its values at $t\in\{0,1,...,K\}$.	
\end{corollary}

\begin{remark}
	As we will discuss in Section \ref{sec:ext}, we can also extend the techniques that we developed for proving Theorem \ref{th:p} to the decentralized setting. The exact rate-memory tradeoff for both the average rate and the peak rate can be fully characterized using these techniques. Besides, the newly proposed decentralized caching scheme for achieving the minimum rates strictly improves the state of the art \cite{maddah-ali13, amiri2016coded}. 
\end{remark}

\begin{remark}
    Prior to this result, there have been several other works on this coded caching problem. Both centralized and decentralized settings have been considered, and many caching schemes using uncoded prefetching were proposed. Several caching schemes have been proposed focusing on minimizing the average communication rates  \cite{niesen13,zhang15,ji2015order,ramakrishnan2015efficient}. However in the case of uniform file popularity, the achievable rates provided in these works reduce to the results of \cite{maddah-ali12a} or \cite{maddah-ali13}, while our proposal strictly improves the state of the arts in both \cite{maddah-ali12a} and \cite{maddah-ali13} by developing a novel delivery strategy that exploits the commonality of the user demands. There have also been several proposed schemes that aim to minimize the peak rates \cite{wan2016caching,amiri2016coded}. The main novelty of our work compared to their results is that we not only propose an optimal design that strictly improves upon all these works through a leader based strategy, but also provide an intuitive  proof for its decodability. The decodability proof is based on the observation that the caching schemes proposed in \cite{maddah-ali12a} and \cite{maddah-ali13} may introduce redundancy in the delivery phase, while our proposed scheme
    provides a systematic way to \emph{optimally} remove all the redundancy, which allows delivering the same amount of information with strictly improved communication rates.
\end{remark}

 	\begin{figure}[htbp]
\centering
\begin{subfigure}{.45\textwidth}
  \centering
  \captionsetup{justification=centering}
  \includegraphics[width=0.95\linewidth]{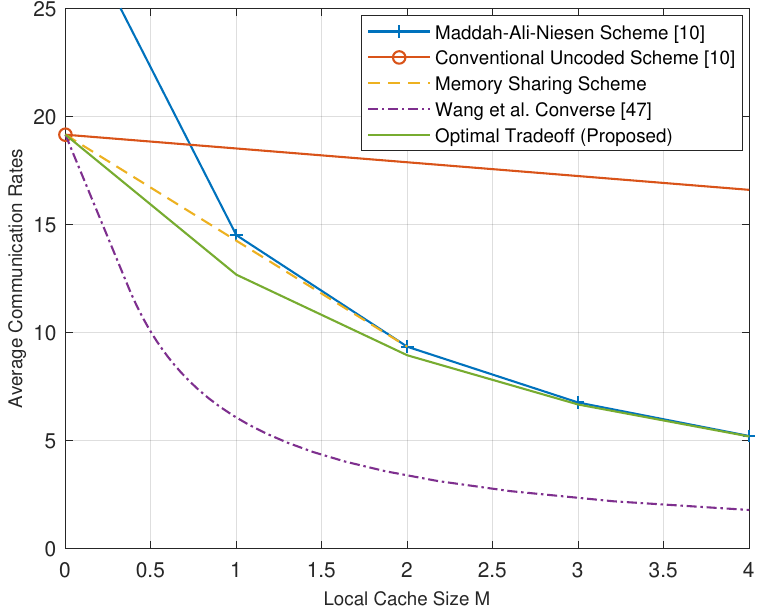}
  \caption{Average rates for $N=K=30$. For this scenario, the best communication rate stated in prior works is achieved by the memory-sharing between the conventional uncoded scheme \cite{maddah-ali12a} and the Maddah-Ali-Niesen scheme \cite{maddah-ali12a}. The tightest prior converse bound in this scenario is provided by \cite{DBLP:journals/corr/WangLG16}. }
\end{subfigure}

 	\vspace{8mm}
\begin{subfigure}{.45\textwidth}
  \centering
  \captionsetup{justification=centering}
  \includegraphics[width=0.95\linewidth]{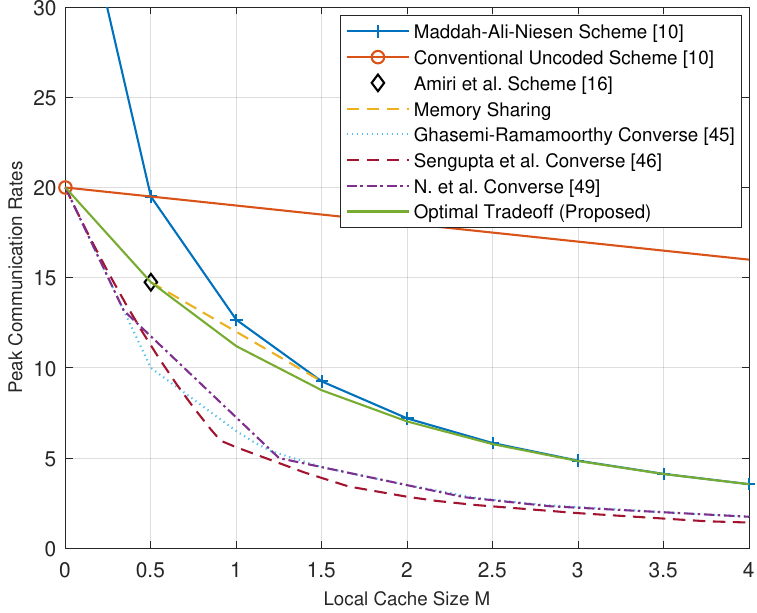}
  \caption{Peak rates for $N=20$, $K=40$. For this scenario, the best communication rate stated in prior works is achieved by the memory-sharing among the conventional uncoded scheme \cite{maddah-ali12a}, the Maddah-Ali-Niesen scheme \cite{maddah-ali12a}, and the Amiri et al. scheme \cite{amiri2016coded}. The tightest prior converse bound in this scenario was provided by \cite{ghasemi15,sengupta15,prem2015critical}. }
\end{subfigure}
\vspace{3mm}
\caption{Numerical comparison between the optimal tradeoff and the state of the arts for the centralized setting. Our results strictly improve the prior arts in both achievability and converse, for both average rate and peak rate.}
\label{fig:cent_compare}
\end{figure}

\begin{remark}
  We numerically compare our results with the state-of-the-art schemes and the converses for the centralized setting. As shown in Fig. \ref{fig:cent_compare}, both the achievability scheme and the converse provided in our paper strictly improve the prior arts, for both average rate and peak rate. Similar results can be shown for the decentralized setting, and a numerical comparison is provided in Section \ref{sec:ext}.

\end{remark}

\begin{remark}
  There have also been several prior works considering caching designs with coded prefetching \cite{maddah-ali12a,DBLP:journals/corr/Chen14h,  sahraei2016k, tian2016caching, amiri2016fundamental}. They focused on the centralized setting and showed that the peak communication rate achieved by uncoded prefetching schemes can be improved in  some low capacity regimes. Even taking coded prefetching schemes into account, our work strictly improves the prior art in most cases (see Section \ref{sec:conclu} for numerical results). 
  More importantly, the caching schemes developed in this paper is within a factor of $2$ optimal in the general coded prefetching setting, for both average and peak rates, centralized and decentralized settings \cite{yu2017characterizing}.

\end{remark}

	In the following sections, we prove Theorem \ref{th:p} by first describing a caching scheme that achieves the minimum average rate (see Section \ref{sec:opt}), and then deriving tight lower bounds of the expected rates for any uncoded prefetching scheme (see Section \ref{sec:conv}).

       	\section{The Optimal Caching Scheme}\label{sec:opt}
	
		In this section, we provide a caching scheme (i.e. a prefetching scheme and a delivery scheme) to achieve $R^*$ stated in Theorem \ref{th:p}.	Before introducing the proposed caching scheme, we demonstrate the main ideas of the proposed scheme through a motivating example.

	\subsection{Motivating Example}\label{sec:example}
	Consider a caching system with $3$ files (denoted by $A$, $B$, and $C$), $6$ users, and a caching size of $1$ file for each user. 	
	To develop a caching scheme, we need to design an uncoded prefetching scheme, independent of the demands, and develop delivery strategies for each of the possible $3^6$ demands.
	
	For the prefetching strategy, we break file $A$ into $15$ subfiles of equal size, and denote their values by $A_{\{1,2\}}$, $A_{\{1,3\}}$, $A_{\{1,4\}}$, $A_{\{1,5\}}$, $A_{\{1,6\}}$, $A_{\{2,3\}}$, $A_{\{2,4\}}$, $A_{\{2,5\}}$, $A_{\{2,6\}}$, $A_{\{3,4\}}$, $A_{\{3,5\}}$, $A_{\{3,6\}}$, $A_{\{4,5\}}$, $A_{\{4,6\}}$, and $A_{\{5,6\}}$. Each user $k$ caches the subfiles whose index includes $k$, e.g., user $1$ caches $A_{\{1,2\}}$, $A_{\{1,3\}}$, $A_{\{1,4\}}$, $A_{\{1,5\}}$, and $A_{\{1,6\}}$. The same goes for files $B$ and $C$. This prefetching scheme was originally proposed in \cite{maddah-ali12a}. 
	
	Given the above prefetching scheme, we now need to develop an optimal delivery strategy for each of the possible demands. In this subsection, we demonstrate the key idea of our proposed delivery scheme through a representative demand scenario, namely, each file is requested by 2 users as shown in Figure \ref{fig:example}.
	
		\begin{figure}[htbp]
			\centering
			\includegraphics[width=0.45\textwidth]{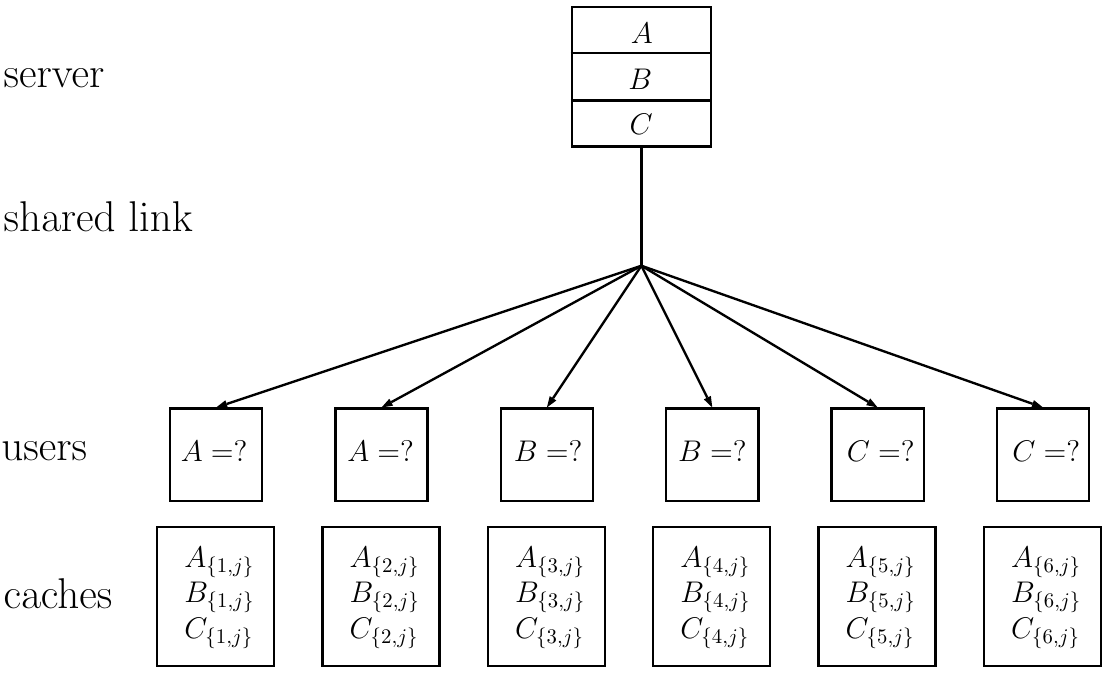}
			\caption{A caching system with $6$ users, $3$ files, local cache size of $1$ file at each user, and a demand where each file is requested by $2$ users.}
			\label{fig:example}
		\end{figure} 
	
	We first consider a subset of 3 users $\{1,2,3\}$. User $1$ requires subfile $A_{\{2,3\}}$, which is only available at users $2$ and $3$. User $2$ requires subfile $A_{\{1,3\}}$, which is only available at users $1$ and $3$. User $3$ requires subfile $B_{\{1,2\}}$, which is only available at users $1$ and $2$. In other words, the three users would like to exchange subfiles $A_{\{2,3\}}$, $A_{\{1,3\}}$, and $B_{\{1,2\}}$, which can be enabled by transmitting the message $A_{\{2,3\}}\oplus A_{\{1,3\}}\oplus B_{\{1,2\}}$ over the shared link.

	Similarly, we can create and broadcast messages for any subset $\mathcal{A}$ of $3$ users that exchange $3$ subfiles among those $3$ users. As a short hand notation, we denote the corresponding message by $Y_\mathcal{A}$. According to the delivery scheme proposed in \cite{maddah-ali12a}, if we broadcast all $\binom{6}{3}=20$ messages that could be created in this way, all users will be able to decode their requested files.
	
	However, in this paper we propose a delivery scheme where, instead of broadcasting all those $20$ messages, only $19$ of them are computed and broadcasted, omitting the message $Y_{\{2,4,6\}}$.  
	Specifically, we broadcast the following $19$ values:
	\begin{align*}
	Y_{\{1,2,3\}}=B_{\{1,2\}}\oplus A_{\{1,3\}} \oplus A_{\{2,3\}} \\ Y_{\{1,2,4\}}=B_{\{1,2\}}\oplus A_{\{1,4\}} \oplus A_{\{2,4\}}\\
	Y_{\{1,2,5\}}=C_{\{1,2\}}\oplus A_{\{1,5\}} \oplus A_{\{2,5\}} \\
	Y_{\{1,2,6\}}=C_{\{1,2\}}\oplus A_{\{1,6\}} \oplus A_{\{2,6\}}\\
	Y_{\{1,3,4\}}=B_{\{1,3\}}\oplus B_{\{1,4\}} \oplus A_{\{3,4\}}\\
	Y_{\{1,3,5\}}=C_{\{1,3\}}\oplus B_{\{1,5\}} \oplus A_{\{3,5\}}\\
	Y_{\{1,3,6\}}=C_{\{1,3\}}\oplus B_{\{1,6\}} \oplus A_{\{3,6\}}\\
	Y_{\{1,4,5\}}=C_{\{1,4\}}\oplus B_{\{1,5\}} \oplus A_{\{4,5\}}\\
	Y_{\{1,4,6\}}=C_{\{1,4\}}\oplus B_{\{1,6\}} \oplus A_{\{4,6\}}\\
	Y_{\{1,5,6\}}=C_{\{1,5\}}\oplus C_{\{1,6\}} \oplus A_{\{5,6\}}\\
	Y_{\{2,3,4\}}=B_{\{2,3\}}\oplus B_{\{2,4\}} \oplus A_{\{3,4\}}\\
	Y_{\{2,3,5\}}=C_{\{2,3\}}\oplus B_{\{2,5\}} \oplus A_{\{3,5\}}\\
	Y_{\{2,3,6\}}=C_{\{2,3\}}\oplus B_{\{2,6\}} \oplus A_{\{3,6\}}\\
	Y_{\{2,4,5\}}=C_{\{2,4\}}\oplus B_{\{2,5\}} \oplus A_{\{4,5\}}\\
	Y_{\{2,5,6\}}=C_{\{2,5\}}\oplus C_{\{2,6\}} \oplus A_{\{5,6\}}\\
	Y_{\{3,4,5\}}=C_{\{3,4\}}\oplus B_{\{3,5\}} \oplus B_{\{4,5\}}\\
	Y_{\{3,4,6\}}=C_{\{3,4\}}\oplus B_{\{3,6\}} \oplus B_{\{4,6\}}\\
	Y_{\{3,5,6\}}=C_{\{3,5\}}\oplus C_{\{3,6\}} \oplus B_{\{5,6\}}\\
	Y_{\{4,5,6\}}=C_{\{4,5\}}\oplus C_{\{4,6\}} \oplus B_{\{5,6\}}
	\end{align*}

	Surprisingly, even after taking out the extra message, all users are still able to decode the requested files. The reason is as follows:
	
	User $1$ is able to decode file $A$, because every subfile $A_{\{i,j\}}$ that is not cached by user $1$ can be computed with the help of $Y_{\{1,i,j\}}$, which is directly broadcasted. 
	The above is the same decoding procedure used in \cite{maddah-ali12a}. 	User $2$ can easily decode all subfiles in $A$ except $A_{\{4,6\}}$ in a similar way, although decoding $A_{\{4,6\}}$ is more challenging since the value $Y_{\{2,4,6\}}$, which is needed in the above decoding procedure for decoding $A_{\{4,6\}}$, is not directly broadcasted. However, user $2$ can still decode $A_{\{4,6\}}$ by adding $Y_{\{1,4,6\}}$, $Y_{\{1,4,5\}}$, $Y_{\{1,3,6\}}$, and $Y_{\{1,3,5\}}$, which gives the binary sum of $A_{\{4,6\}}$, $A_{\{4,5\}}$, $A_{\{3,6\}}$, and $A_{\{3,5\}}$. Because $A_{\{4,5\}}$, $A_{\{3,6\}}$, and $A_{\{3,5\}}$ are easily decodable, $A_{\{4,6\}}$ can be obtained consequently.

	Due to symmetry, all other users can decode their requested files in the same manner. This completes the decoding tasks for the given demand.

	\subsection{General Schemes}
	
	Now we present a general caching scheme that achieves the rate $R^*$ stated in Theorem \ref{th:p}. We focus on presenting prefetching schemes and delivery schemes when $t\in\{0,1,...,K\}$, since for general $t$, the minimum rate $R^*$ can be achieved by memory sharing.

	\begin{remark}\label{remark:convexity}
		Note that the rates stated in equation (\ref{eq:unif_pre}) for $t\in\{0,1,...,K\}$ form a convex sequence, which are consequently on their lower convex envelope. Thus those rates cannot be further improved using memory sharing.
	\end{remark}
	
	To prove the achievability of $R^*$, we need to provide an optimal prefetching scheme $\boldsymbol{\mathcal{M}}$, an optimal delivery scheme for every possible user demand $\boldsymbol{d}$ of which the average rate achieves $R^*$, and a valid decoding algorithm for the users. The main idea of our proposed achievability scheme is to first design a prefetching scheme that enables multicast coding opportunities, and then in the delivery phase, we optimally deliver the message by effectively solving an index coding problem. 
	
	We consider the following optimal prefetching:	We partition each file $i$ into $\binom{K}{t}$ non-overlapping subfiles with approximately equal size. We assign the $\binom{K}{t}$ subfiles to $\binom{K}{t}$ different subsets of $\{1,..,K\}$ of size $t$, and denote the value of the subfile assigned to subset $\mathcal{A}$ by $W_{i,\mathcal{A}}$. 
		Given this partition, each user $k$ caches all bits in all subfiles $W_{i,\mathcal{A}}$ such that $k\in\mathcal{A}$. Because each user caches $\binom{K-1}{t-1}N$ subfiles, and each subfile has $F/\binom{K}{t}$ bits, the caching load of each user equals $NtF/K=MF$ bits, which satisfies the memory constraint. This prefetching was originally proposed in \cite{maddah-ali12a}. In the rest of the paper, we refer to this prefetching as \textit{symmetric batch prefetching}.
			
			\
		Given this prefetching (denoted by $\boldsymbol{\mathcal{M_{\textup{batch}}}}$), our goal is to show that for any demand $\boldsymbol{d}$, we can find a delivery scheme that achieves the following optimal rate with zero error probability: \footnote{Rigorously, we prove equation (\ref{eq:single}) for  $F|\binom{K}{t}$. In other cases, the resulting extra communication overhead is negligible for large $F$.}
		 	\begin{align}\label{eq:single}
		 	R^*_{\epsilon=0}\left(\boldsymbol{d},\boldsymbol{\mathcal{M_{\textup{batch}}}}\right)=\frac{\binom{K}{t+1}-\binom{K-N_{\textup{e}}(\boldsymbol{d})}{t+1}}{\binom{K}{t}}.
		 	\end{align}
		Hence, by taking the expectation over demand $\boldsymbol{d}$, the rate $R^*$ stated in Theorem \ref{th:p} can be achieved. 
	
	\begin{remark}
			Note that, in the special case where all users are requesting different files (i.e., $N_{\textup{e}}(\boldsymbol{d})=K$), the above rate equals $\frac{K-t}{t+1}$, which can already be achieved by the delivery scheme proposed in \cite{maddah-ali12a}. Our proposed scheme aims to achieve this optimal rate in more general circumstances, when some users may share common demands.
	\end{remark}
	
		\begin{remark}
			Finding the minimum communication load given a prefetching $\boldsymbol{\mathcal{M}}$ can be viewed as a special case of the index coding problem. Theorem \ref{th:p} indicates the optimality of the delivery scheme given the symmetric batch prefetching, which implies that (\ref{eq:single}) gives the solution to a special class of non-symmetric index coding problem. 
		\end{remark}
	
	The optimal delivery scheme is designed as follows:
	For each demand $\boldsymbol{d}$, recall that $N_{\textup{e}}(\boldsymbol{d})$ denotes the number of distinct files requested by all users. The server arbitrarily selects a subset of $N_{\textup{e}}(\boldsymbol{d})$ users, denoted by $\mathcal{U}=\{u_1,...,u_{N_{\textup{e}}(\boldsymbol{d})}\}$, that request $N_{\textup{e}}(\boldsymbol{d})$ different files. We refer to these users as \textit{leaders}. 
	
	Given an arbitrary subset $\mathcal{A}$ of $t+1$ users, each user $k\in \mathcal{A}$ needs the subfile $W_{d_k, \mathcal{A}\backslash\{k\}}$, which is known by all other users in $\mathcal{A}$.
	In other words, all users in set $\mathcal{A}$ would like to exchange subfiles $W_{d_k, \mathcal{A}\backslash\{k\}}$ for all $k\in \mathcal{A}$. This exchange can be processed if the binary sum of all those files, i.e. $\underset{x\in \mathcal{A}}{\mathlarger{\oplus}}W_{d_x, \mathcal{A}\backslash\{x\}}$, is available from the broadcasted message.  To simplify the description of the delivery scheme, for each subset $\mathcal{A}$ of users, we define the following short hand notation
	\begin{align}
	Y_\mathcal{A}=\underset{x\in \mathcal{A}}{\mathlarger{\oplus}}W_{d_x, \mathcal{A}\backslash\{x\}}.\label{eq:ya}
	\end{align}
	
	To achieve the rate stated in (\ref{eq:single}), the
	server only greedily broadcasts the binary sums that directly help at least $1$ leader. Rigorously, the server computes and broadcasts all $Y_\mathcal{A}$ for all subsets $\mathcal{A}$ of size $t+1$ that satisfy $\mathcal{A}\cap\mathcal{U}\neq\emptyset$. The length of the message equals $\binom{K}{t+1}-\binom{K-N_{\textup{e}}(\boldsymbol{d})}{t+1}$ times the size of a subfile, which matches the stated rate.
	
	We now prove that each user who requests a file is able to decode the requested file upon receiving the messages.	For any leader $k\in \mathcal{U}$ and any subfile $W_{d_k, \mathcal{A}}$ that is requested but not cached by user $k$, the message $Y_{\{k\}\cup \mathcal{A}}$ is directly available from the broadcast. Thus, $k$ is able to obtain all requested subfiles by decoding each subfile  $W_{d_k, \mathcal{A}}$ from message $Y_{\{k\}\cup \mathcal{A}}$ using the following equation: 
			\begin{align}
	W_{d_k, \mathcal{A}}=  Y_{\{k\}\cup \mathcal{A}} \  \mathlarger{\oplus} \left(\underset{x\in \mathcal{A}}{\mathlarger{\oplus}}W_{d_x, \{k\}\cup \mathcal{A}\backslash\{x\}}\right),
	\end{align}
	which directly follows from equation (\ref{eq:ya}).
	
	The decoding procedure for a non-leader user $k$ is less straightforward, because not all messages $Y_{\{k\}\cup\mathcal{A}}$ for corresponding required subfiles $W_{d_k, \mathcal{A}}$ are directly broadcasted. However, user $k$ can generate these messages simply based on the received messages, and can thus decode all required subfiles. We prove the above fact as follows.

	First we prove the following simple lemma:
	
	\begin{lemma}
		\label{dec_l}
		Given a demand $\boldsymbol{d}$, and a set of leaders $\mathcal{U}$. For any subset $\mathcal{B}\subseteq\{1,...,K\}$ that includes $\mathcal{U}$, let $\mathcal{V}_{\textup{F}}$ be the family of all subsets $\mathcal{V}$ of $\mathcal{B}$ such that each requested file in $\boldsymbol{d}$ is requested by exactly one user in $\mathcal{V}$.
		
		The following equation holds: 
		\begin{align}
		\label{dec_lemma}
		\underset{\mathcal{V}\in\mathcal{V}_{\textup{F}}}{\mathlarger{\oplus}} Y_{\mathcal{B}\backslash\mathcal{V}}=0
		\end{align}
		if each $Y_{\mathcal{B}\backslash\mathcal{V}}$ is defined in (\ref{eq:ya}).
		
	\end{lemma}
	\begin{proof}
		To prove Lemma \ref{dec_l}, we essentially need to show that, after expanding the LHS of equation (\ref{dec_lemma}) into a binary sum of subfiles using the definition in (\ref{eq:ya}), each subfile is counted an even number of times. This will ensure that the net sum is equal to $0$. To rigorously prove this fact, we start by defining the following.

		For each $u\in \mathcal{U}$ we define $\mathcal{B}_u$ as 
		\begin{align}
		\mathcal{B}_u=\{x\in \mathcal{B}\ |\ d_x=d_u\}.
		\end{align}

		Then all sets $\mathcal{B}_u$ disjointly cover the set $ \mathcal{B}$, and the following equations hold:
		
		\begin{align}
		\underset{\mathcal{V}\in\mathcal{V}_{\textup{F}}}{\mathlarger{\oplus}} Y_{ \mathcal{B}\backslash\mathcal{V}}&=\underset{\mathcal{V}\in\mathcal{V}_{\textup{F}}}{\mathlarger{\oplus}} \ \underset{x\in  \mathcal{B}\backslash\mathcal{V}}{\mathlarger{\oplus}}W_{d_x,  \mathcal{B}\backslash(\mathcal{V} \cup \{x\})}\\
		&=\underset{u\in\mathcal{U}}{\mathlarger{\oplus}} \ \underset{\mathcal{V}\in\mathcal{V}_{\textup{F}}}{\mathlarger{\oplus}} \ \underset{x\in ( \mathcal{B}\backslash\mathcal{V})\cap \mathcal{B}_u}{\mathlarger{\oplus}}W_{d_u,  \mathcal{B}\backslash(\mathcal{V} \cup \{x\})}\\
		&=\underset{u\in\mathcal{U}}{\mathlarger{\oplus}} \ \underset{\mathcal{V}\in\mathcal{V}_{\textup{F}}}{\mathlarger{\oplus}} \ \underset{x\in \mathcal{B}_u\backslash\mathcal{V}}{\mathlarger{\oplus}}W_{d_u, \mathcal{B}\backslash(\mathcal{V} \cup \{x\})}.
		\end{align}
		
		For each $u\in\mathcal{U}$, we let $\mathcal{V}_u$ be the family of  all subsets $\mathcal{V}'$ of $\mathcal{B}\backslash \mathcal{B}_u$ such that each requested file in $\boldsymbol{d}$, except $d_u$, is requested by exactly one user in $\mathcal{V}'$.	Then $\mathcal{V}_{\textup{F}}$ can be represented as follows:
		\begin{align}
		\mathcal{V}_{\textup{F}}=\{ \{y\}\cup \mathcal{V}'\ |\  y\in\mathcal{B}_u , \mathcal{V}'\in\mathcal{V}_u\}. 
		\end{align}
		
		Consequently, the following equation holds for each $u\in\mathcal{U}$:
		\begin{align}
		\underset{\mathcal{V}\in\mathcal{V}_{\textup{F}}}{\mathlarger{\oplus}} \ \underset{x\in \mathcal{B}_u\backslash\mathcal{V}}{\mathlarger{\oplus}}&W_{d_u, \mathcal{B}\backslash(\mathcal{V} \cup \{x\})}\nonumber\\&=\underset{\mathcal{V}'\in\mathcal{V}_u}{\mathlarger{\oplus}}  
		\  \underset{y\in\mathcal{B}_u}{\mathlarger{\oplus}} \ \underset{x\in \mathcal{B}_u \backslash\{y\} }{\mathlarger{\oplus}}W_{d_u, \mathcal{B}\backslash(\mathcal{V}' \cup \{x,y\})}\\
		&=\underset{\mathcal{V}'\in\mathcal{V}_u}{\mathlarger{\oplus}}  
		\  \underset{\substack{(x,y)\in\mathcal{B}^2_u\\x\neq y}}{\mathlarger{\oplus}} W_{d_u, \mathcal{B}\backslash(\mathcal{V}' \cup \{x,y\})}
		\end{align}
		Note that $W_{d_u, \mathcal{B}\backslash(\mathcal{V}' \cup \{x,y\})}$ and $W_{d_u, \mathcal{B}\backslash(\mathcal{V}' \cup \{y,x\})}$ are the same subfile. Hence, every single subfile in the above equation is counted exactly twice, which sum up to $0$.
		Consequently, the LHS of equation (\ref{dec_lemma}) also equals 0.
	\end{proof}
	
	Consider any subset $\mathcal{A}$ of $t+1$ non-leader users. From Lemma \ref{dec_l}, the message $Y_\mathcal{A}$ can be directly computed from the broadcasted messages using the following equation: 
	\begin{align}
	Y_\mathcal{A}=\underset{\mathcal{V}\in\mathcal{V}_{\textup{F}}\backslash \{\mathcal{U}\}}{\mathlarger{\oplus}} Y_{\mathcal{B}\backslash\mathcal{V}},
	\end{align}
	where $\mathcal{B}=\mathcal{A}\cup \mathcal{U}$, 
	given the fact that all messages on the RHS of the above equation are broadcasted, because each $\mathcal{B}\backslash\mathcal{V}$ has a size of $t+1$ and contains at least one leader.	Hence, each user $k$ can obtain the value $Y_\mathcal{A}$ for any subset $\mathcal{A}$ of $t+1$ users, and can subsequently decode its requested file as previously discussed.

\begin{remark}
	An interesting open problem is to find computationally efficient decoding algorithms for the proposed optimal caching scheme. 
	The decoding algorithm proposed in this paper imposes extra computation
	at the non-leader users, since they have to solve for the missing messages to recover all needed subfiles. However there are some ideas that one may explore to improve this decoding strategy, e.g. designing a smarter approach for non-leader users instead of naively recovering all required messages before decoding the subfiles (see the decoding approach provided in the motivating example in Section \ref{sec:example}). 
	\end{remark}

	\section{Converse}\label{sec:conv}

	In this section, we derive a tight lower bound on the minimum expected rate $R^*$, which shows the optimality of the caching scheme proposed in this paper. 
	To derive the corresponding lower bound on the average rate over all demands, we divide the set $\mathcal{D}$ into smaller subsets, and lower bound the average rates within each subset individually. We refer to these smaller subsets as \textit{types}, which are defined as follows.\footnote{The notion of type was also recently introduced in \cite{chao2016isit} in order to simplify the LP for finding better converse bounds for the coded caching problem.}

	 	\begin{figure}[htbp]
	 		\centering
	 		\includegraphics[width=0.35\textwidth]{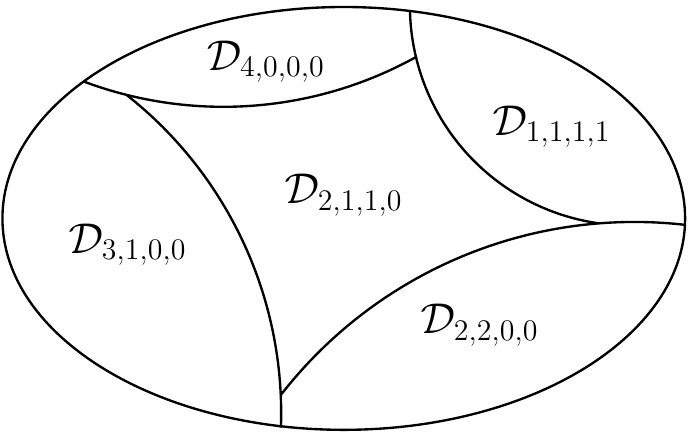}
	 		\caption{Dividing $\mathcal{D}$ into 5 types, for a caching problem with $4$ files and $4$ users.}
	 		\label{fig:types}
	 	\end{figure}
	 Given an arbitrary demand $\boldsymbol{d}$, we define its \textit{statistics}, denoted by  $\boldsymbol{s}(\boldsymbol{d})$, as a sorted array of length $N$, such that $s_i(\boldsymbol{d})$ equals the number of users that request the $i$th most requested file. We denote the set of all possible statistics by $\mathcal{S}$. 
	 Grouping by the same statistics, the set of all demands $\mathcal{D}$ can be broken into many small subsets.	For any statistics $\boldsymbol{s}\in\mathcal{S}$, we define type $\mathcal{D}_{\boldsymbol{s}}$ as the set of queries with statistics $\boldsymbol{s}$. 
	 
	 For example, consider a caching problem with $4$ files (denoted by $A$, $B$, $C$, and $D$) and $4$ users. The statistics of the demand $\boldsymbol{d}=(A,A,B,C)$ equals $\boldsymbol{s}(\boldsymbol{d})=(2,1,1,0)$. More generally, the set of all possible statistics for this problem is $\mathcal{S}=\{(4,0,0,0),(3,1,0,0),(2,2,0,0),(2,1,1,0),(1,1,1,1)\}$, and $\mathcal{D}$ can be divided into $5$ types accordingly, as shown in Fig. \ref{fig:types}.
	 
	 Note that for each demand  $\boldsymbol{d}$, the value $N_{\textup{e}}(\boldsymbol{d})$ only depends on its statistics $\boldsymbol{s}({\boldsymbol{d}})$, and thus the value is identical across all demands in $\mathcal{D}_{\boldsymbol{s}}$. For convenience, we denote that value by $N_{\textup{e}}(\boldsymbol{s})$.
	 
	 Given a prefetching $\boldsymbol{\mathcal{M}}$, we denote the average rate within each type $\mathcal{D}_{\boldsymbol{s}}$ by $R^*_{\epsilon}(\boldsymbol{s}, \boldsymbol{\mathcal{M}})$. Rigorously, 
	 \begin{align}
	 R^*_\epsilon({\boldsymbol{s}}, \boldsymbol{\mathcal{M}})&=\frac{1}{|\mathcal{D}_{\boldsymbol{s}}|}\sum_{\boldsymbol{d}\in \mathcal{D}_{\boldsymbol{s}}}R^*_\epsilon(\boldsymbol{d},\boldsymbol{\mathcal{M}}).
	 \end{align}
	 Recall that all demands are equally likely, so we have
	  \begin{align}
	  R^*&=\sup_{\epsilon>0} \adjustlimits\limsup_{F\rightarrow+\infty}\min_{{\boldsymbol{\mathcal{M}}}} \mathbb{E}_{\boldsymbol{s}}[R_\epsilon^*({\boldsymbol{s}}, \boldsymbol{\mathcal{M}})]\\
	  &\geq \sup_{\epsilon>0} \limsup_{F\rightarrow+\infty}\mathbb{E}_{\boldsymbol{s}}[\min_{{\boldsymbol{\mathcal{M}}}} R_\epsilon^*({\boldsymbol{s}}, \boldsymbol{\mathcal{M}})].\label{ineq:17}
	  \end{align}
	  Hence, in order to lower bound $R^*$,  
	  it is sufficient to bound the minimum value of $R^*_\epsilon({\boldsymbol{s}}, \boldsymbol{\mathcal{M}})$ for each type $\mathcal{D}_{\boldsymbol{s}}$ individually.
	  We show that, when the prefetching is uncoded, the minimum average rate within a type can be tightly bounded (when $F$ is large and $\epsilon$ is small), thus the rate-memory tradeoff can be completely characterized using this technique.
	  
	  The lower bounds of the minimum average rates within each type are presented in the following lemma:
	 \begin{lemma}\label{lemma:univ}
	 	Consider a caching problem with $N$ files, $K$ users, and a local cache size of $M$ files for each user. For any type $\mathcal{D}_{\boldsymbol{s}}$, the minimum value of $R_\epsilon^*(\boldsymbol{s}, \boldsymbol{\mathcal{M}})$ is lower bounded by
	 	\begin{align}
	 	\min_{\boldsymbol{\mathcal{M}}}	
	 	R^*_\epsilon(\boldsymbol{s} ,\boldsymbol{\mathcal{M}})\geq &\textup{Conv}\left( \frac{\binom{K}{t+1}-\binom{K-N_{\textup{e}}(\boldsymbol{s})}{t+1}}{\binom{K}{t}}\right)\nonumber\\& -\left(\frac{1}{F}+ N^2_{\textup{e}}(\boldsymbol{s})\epsilon  \right),
	 	 	\end{align}
	 	where $\textup{Conv}(f(t))$ denotes the lower convex envelope of the following points: $\{(t,f(t))\ | \ t\in\{0,1,...,K\}\}$.
	 \end{lemma}
	 
	 \begin{remark}[Universal Optimality of Symmetric Batch Prefetching]
		 	The above lemma characterizes the minimum average rate given a type $\mathcal{D}_{\boldsymbol{s}}$, if the prefetching $\boldsymbol{\mathcal{M}}$ can be designed based on $s$. However,
		 	for (\ref{ineq:17}) to be tight, the average rate for each different type has to be minimized on the same prefetching. Surprisingly, such an optimal prefetching exists, an example being the symmetric batch prefetching according to Section \ref{sec:opt}. This indicates that the symmetric batch prefetching is universally optimal for all types in terms of the average rates.
	 \end{remark}

	 We postpone the proof of Lemma \ref{lemma:univ} to Appendix \ref{app:keylemma} and first prove the converse using the lemma.
	
	From (\ref{ineq:17}) and Lemma \ref{lemma:univ}, $R^*$ can be lower bounded as follows:
	\begin{align}
	R^*&\geq\sup_{\epsilon>0} \limsup_{F\rightarrow+\infty}\mathbb{E}_{\boldsymbol{s}}\left[\min_{\boldsymbol{\mathcal{M}}}R_\epsilon^*({\boldsymbol{s}},\boldsymbol{\mathcal{M}})\right] \\&\geq \mathbb{E}_{\boldsymbol{s}}\left[\textup{Conv}\left(\frac{\binom{K}{t+1}-\binom{K-N_{\textup{e}}(\boldsymbol{s})}{t+1}}{\binom{K}{t}}\right)\right].\label{eq:39}
	\end{align}
	Because the sequence 
	\begin{align}
	c_n &= \frac{\binom{K}{n+1}-\binom{K-N_{\textup{e}}(\boldsymbol{s})}{n+1}}{\binom{K}{n}} 
	\end{align}
	is convex, we can switch the order of the expectation and the Conv in (\ref{eq:39}).
		Therefore,  $R^*$ is lower bounded by the rate defined in Theorem \ref{th:p}.\footnote{As noted in Remark \ref{remark:convexity}, the rate $R^*$ stated in equation (\ref{eq:unif_pre}) for $t\in\{0,1,...,K\}$ is convex, so it is sufficient to prove $R^*$ is lower bounded by the convex envelope of its values at $t\in\{0,1,...,K\}$.}

	\section{Extension to the Decentralized Setting}\label{sec:ext}
	In the sections above, we introduced a new centralized caching scheme and a new bounding technique that completely characterize the minimum average communication rate and the minimum peak rate, when the prefetching is required to be uncoded. Interestingly, these techniques can also be extended to fully characterize the rate-memory tradeoff for decentralized caching. In this section, we formally establish a system model for decentralized caching systems, and state the exact rate-memory tradeoff as main results for both the average rate and the peak rate.

	\subsection{System Model and Problem Formulation}

	In many practical systems, out of the large number of users that may potentially request files from the server through the shared error-free link, only a random unknown subset are connected to the link and making requests at any given time instance. To handle this situation, the concept of decentralized prefetching scheme was introduced in \cite{maddah-ali13}, where each user has to fill their caches randomly and independently, based on the same probability distribution. The goal in the decentralized setting is to find a decentralized prefetching scheme, without the knowledge of the number and the identities of the users making requests, to minimize the required communication rates given an arbitrarily large caching system.  
	Based on the above framework, we formally define decentralized caching as follows:
	\begin{definition}
		In a \textit{decentralized caching scheme}, instead of following a deterministic caching scheme, each user $k$ caches a subset $\mathcal{M}_k$ of size no more than $MF$ bits randomly and independently, based on the same probability distribution, denoted by $P_{{\mathcal{M}}}$. Rigorously, when $K$ users are making requests, the probability distribution of the prefetching $\boldsymbol{\mathcal{M}}$ is given by
		\begin{align}
		    \mathbbm{P}\left(\boldsymbol{\mathcal{M}}=(\mathcal{M}_1,...,\mathcal{M}_K)\right)= \prod_{i=1}^{K}{{P}_{\mathcal{M}}(\mathcal{M}_i)}.\nonumber
		\end{align}
		We define that a \textit{decentralized caching scheme}, denoted by ${P}_{\mathcal{M};F}$, is a distribution parameterized by the file size $F$, that specifies the prefetching distribution $P_{\mathcal{M}}$ for all possible values of $F$. 
	\end{definition}

	Similar to the centralized setting, when $K$ users are making requests, we say that a rate $R$ is \textit{$\epsilon$-achievable} given a prefetching distribution $P_{{\mathcal{M}}}$ and a demand $\boldsymbol{d}$ if and only if there exists a message $X$ of length $RF$ such that every active user $k$ is able to recover its desired file $W_{d_k}$ with a probability of error of at most $\epsilon$. This is rigorously defined as follows:  
		
		\begin{definition}
		 When $K$ users are making requests, $R$ is \textit{$\epsilon$-achievable} given a prefetching distribution $P_{{\mathcal{M}}}$ and a demand $\boldsymbol{d}$ if and only if for every possible realization of the prefetching $\boldsymbol{\mathcal{M}}$, we can find a real number $\epsilon_{\boldsymbol{\mathcal{M}}}$, such that $R$ is $\epsilon_{\boldsymbol{\mathcal{M}}}$-achievable given  $\boldsymbol{\mathcal{M}}$ and $\boldsymbol{d}$, and
		   $ \mathbbm{E}[\epsilon_{\boldsymbol{\mathcal{M}}}]\leq \epsilon $.
		\end{definition}
	 We denote $R^*_{\epsilon, K}(\boldsymbol{d}, P_{\mathcal{M}})$ as the minimum $\epsilon$-achievable rate given $K$, $\boldsymbol{d}$ and $P_{{\mathcal{M}}}$, and we define the rate-memory tradeoff for the average rate based on this notation. For each $K\in \mathbb{N}$, and each prefetching scheme $P_{\mathcal{M};F}$, we define the minimum average rate $R^*_{K} ({P}_{\mathcal{M};F})$ as the minimum expected rate under uniformly random demand that can be achieved with vanishing error probability for sufficiently large file size. Specifically,
	 \begin{align}
	 R^*_{K} ({P}_{\mathcal{M};F})=\sup_{\epsilon>0} \limsup_{F'\rightarrow+\infty}\mathbb{E}_{\boldsymbol{d}} [R^*_{\epsilon,K}(\boldsymbol{d}, P_{{\mathcal{M};F}}(F=F'))],\nonumber
	 \end{align}
	 where the demand $\boldsymbol{d}$ is uniformly distributed on $\{1,\dots,N\}^K$.

	 	Given the fact that a decentralized prefetching scheme is designed without the knowledge of the number of active users $K$, we characterize the rate-memory tradeoff using an infinite dimensional vector, denoted by 
	 	$\{R_K\}_{K\in\mathbb{N}}$, where each term $R_{K}$ corresponds to the needed communication rates when $K$ users are making requests. We aim to find the region in this infinite dimensional vector space that can be achieved by any decentralized prefetching scheme, and
	 	we denote this region by $\mathcal{R}$.  Rigorously, we aim to find
		\begin{align}
	\mathcal{R}= \underset{P_{{\mathcal{M};F}}}{\cup} \{\{R_K\}_{K\in\mathbb{N}} \ |\  \forall K\in\mathbb{N},  R_K\geq R^*_{K}(P_{{\mathcal{M};F}}) \}\nonumber,
	\end{align}
   which is a function of $N$ and $M$.

	Similarly, we 
	 define the rate-memory tradeoff for the peak rate as follows: For each $K\in \mathbb{N}$, and each prefetching scheme $P_{\mathcal{M};F}$, we define the minimum peak rate $R^*_{K,\textup{peak}} ({P}_{\mathcal{M};F})$ as the minimum  communication rate that can be achieved with vanishing error probability for sufficiently large file size, for the worst case demand. Specifically,
	 \begin{align}
	 R^*_{K,\textup{peak}} ({P}_{\mathcal{M};F})=\sup_{\epsilon>0}\adjustlimits \limsup_{F'\rightarrow+\infty}\max_{\boldsymbol{d}\in \mathcal{D}} [R^*_{\epsilon,K}(\boldsymbol{d}, P_{{\mathcal{M};F}}(F=F'))],\nonumber
	 \end{align}
	 
	We aim to find the region in the infinite dimensional vector space that can be achieved by any decentralized prefetching scheme in terms of the peak rate, and we denote this region  $\mathcal{R}_{\textup{peak}}$. Rigorously, we aim to find 
	\begin{align}
	\mathcal{R}_{\textup{peak}}= \underset{P_{{\mathcal{M};F}}}{\cup}\{\{R_K\}_{K\in\mathbb{N}} \ |\ \forall K\in\mathbb{N},  R_K\geq R^*_{K,\textup{peak}}(P_{{\mathcal{M};F}}) \}\nonumber,
	\end{align}
	as a function of $N$ and $M$.

		\subsection{Exact Rate-Memory Tradeoff for Decentralized Setting}
		
		The following theorem completely characterizes the rate-memory tradeoff for the average rate in the decentralized setting:
		\begin{theorem} 
			\label{thm:ave-dec}
			For a decentralized caching problem with parameters $N$ and $M$, $\mathcal{R}$ is completely characterized by the following equation:
			\begin{align}
			\mathcal{R}=&\left\{\{R_K\}_{K\in\mathbb{N}}\
			\vphantom{\left. R_K\geq \mathbb{E}_{\boldsymbol{d}}\left[\frac{N-M}{M} \left(1-\left(\frac{N-M}{N}\right)^{N_{\textup{e}}(\boldsymbol{d})}\right)\right]\right\}}    
			\right|\nonumber\\ 
			&\ \  \left. R_K\geq \mathbb{E}_{\boldsymbol{d}}\left[\frac{N-M}{M} \left(1-\left(\frac{N-M}{N}\right)^{N_{\textup{e}}(\boldsymbol{d})}\right)\right]\right\},
			\end{align}
			where demand $\boldsymbol{d}$ given each $K$ is uniformly distributed on $\{1,...,N\}^K$
			 and $N_{\textup{e}}(\boldsymbol{d})$ denotes the number of distinct requests in $\boldsymbol{d}$.\footnote{If $M=0$, $\mathcal{R}=\{\{R_K\}_{K\in\mathbb{N}}\ |\  R_K\geq \mathbb{E}_{\boldsymbol{d}}[{N_{\textup{e}}(\boldsymbol{d})}]\}$.}
		\end{theorem}
		The proof of the above theorem is provided in Appendix \ref{app:decent-ave}.
		
		\begin{remark}
			Theorem \ref{thm:ave-dec} demonstrates that $\mathcal{R}$ has a very simple shape with one dominating point: $\{R_K=\mathbb{E}_{\boldsymbol{d}}[\frac{N-M}{M} (1-(\frac{N-M}{N})^{N_{\textup{e}}(\boldsymbol{d})})]\}_{K\in\mathbb{N}}$. In other words, we can find a decentralized prefetching scheme that simultaneously achieves the minimum expected rates for all possible numbers of active users. Therefore, there is no tension among the expected rates for different numbers of active users. In Appendix \ref{app:decent-ave}, we will show that one example of the optimal prefetching scheme is to let each user cache $\frac{MF}{N}$ bits in each file uniformly independently.  
		\end{remark}
		
			\begin{remark}
				To prove Theorem \ref{thm:ave-dec}, we propose a decentralized caching scheme that strictly improves the state of the art \cite{maddah-ali13, amiri2016coded} (see Appendix \ref{app:decent-opt}), for both the average rate and the peak rate. In particular for the average rate, the state-of-the-art scheme proposed in \cite{maddah-ali13} achieves the rate $\frac{N-M}{N}\cdot\min\{\frac{N}{M}(1-(1-\frac{M}{N})^K),\mathbb{E}_{\boldsymbol{d}}[N_e(\boldsymbol{d})]\}$, which is strictly larger than the rate achieved by our proposed scheme $\mathbb{E}_{\boldsymbol{d}}[\frac{N-M}{M}(1-(\frac{N-M}{N})^{N_{\textup{e}}(\boldsymbol{d})})]$ in most cases. Similarly one can show that our scheme strictly improves \cite{amiri2016coded}, and we omit the details for brevity.
			\end{remark}

		 \begin{remark}\label{remark:fund-ave}

		 			We also prove a matching information-theoretic outer bound of $\mathcal{R}$, by showing that the achievable rate of any decentralized caching scheme can be lower bounded by the achievable rate of a caching scheme with centralized prefetching that is used on a system where, there are a large number of users that may potentially request a file, but only a subset of $K$ users are actually making the request.
		 	Interestingly, the tightness of this bound indicates that, in a system where the number of potential users is significantly larger than the number of active users, our proposed decentralized caching scheme is optimal, even compared to schemes where the users are not caching according to an i.i.d..
		 	
		 \end{remark}
		
	Using the proposed decentralized caching scheme and the same converse bounding technique, the following corollary, which completely characterizes the rate-memory tradeoff for the peak rate in the decentralized setting, directly follows:
		\begin{corollary} 
			\label{cor:dec}
			For a decentralized caching problem with parameters $N$ and $M$, the achievable region $\mathcal{R}_{\textup{peak}}$ is completely characterized by the following equation:\footnote{If $M=0$, $\mathcal{R}=\{\{R_K\}_{K\in\mathbb{N}}\ |\  R_K\geq \min\{N,K\}\}$.}
			\begin{align}
			\mathcal{R}_{\textup{peak}}=&\left\{\{R_K\}_{K\in\mathbb{N}}\ 
			\vphantom{\left. R_K\geq \frac{N-M}{M} \left(1-\left(\frac{N-M}{N}\right)^{\min\{N,K\}}\right)\right\}}
			\right|\nonumber \\
			&\ \  \left. R_K\geq \frac{N-M}{M} \left(1-\left(\frac{N-M}{N}\right)^{\min\{N,K\}}\right)\right\}.
			\end{align}
		\end{corollary}
		The proof of the above corollary is provided in Appendix \ref{app:decent-worst}.
	
		\begin{remark}
			Corollary \ref{cor:dec} demonstrates that $\mathcal{R}_{\textup{peak}}$ has a very simple shape with one dominating point: $\{R_K=\frac{N-M}{M} (1-(\frac{N-M}{N})^{\min\{N,K\}})\}_{K\in\mathbb{N}}$. In other words, we can find a decentralized prefetching scheme that simultaneously achieves the minimum peak rates for all possible numbers of active users. Therefore, there is no tension among the peak rates for different numbers of active users. In Appendix \ref{app:decent-worst}, we will show that one example of the optimal prefetching scheme is to let each user cache $\frac{MF}{N}$ bits in each file uniformly independently.  
		\end{remark}
		
		\begin{remark}\label{remark:fund-w}
		
		Similar to the average rate case, a matching converse can be proved by deriving the minimum achievable rates of centralized caching schemes in 			
			a system where only a subset of users are actually making the request. 
			Consequently, 
			in a caching system where the number of potential users is significantly larger than the number of active users, our proposed decentralized scheme is also optimal in terms of peak rate, even compared to schemes where the users are not caching according to an i.i.d..
		\end{remark}

 	\begin{figure}[htbp]
\centering
\begin{subfigure}{.45\textwidth}
  \centering
  \captionsetup{justification=centering}
  \includegraphics[width=0.95\linewidth]{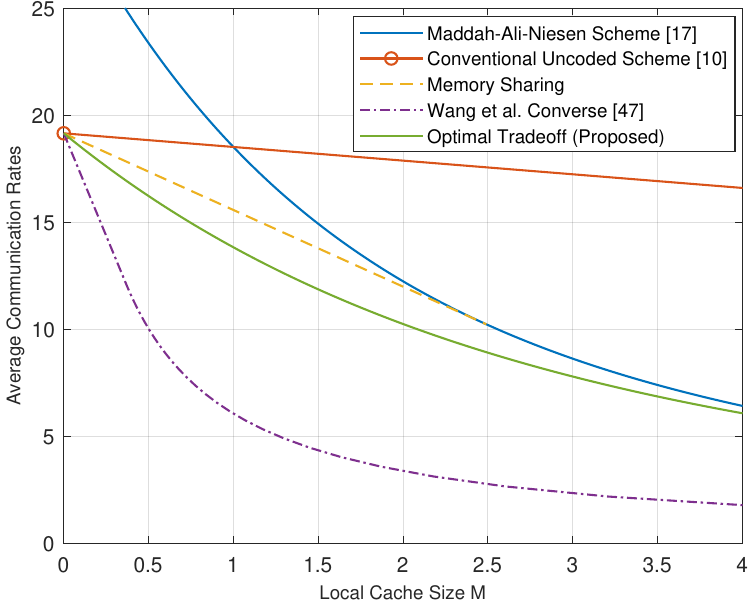}
  \caption{Average rates for $N=K=30$. For this scenario, the best communication rate stated in prior works is achieved by the memory-sharing between the conventional uncoded scheme \cite{maddah-ali12a} and the Maddah-Ali-Niesen scheme \cite{maddah-ali13}. The tightest prior converse bound in this scenario is provided by \cite{DBLP:journals/corr/WangLG16}. }
\end{subfigure}
\vspace{5mm}

\begin{subfigure}{.45\textwidth}
  \centering
  \captionsetup{justification=centering}
  \includegraphics[width=0.95\linewidth]{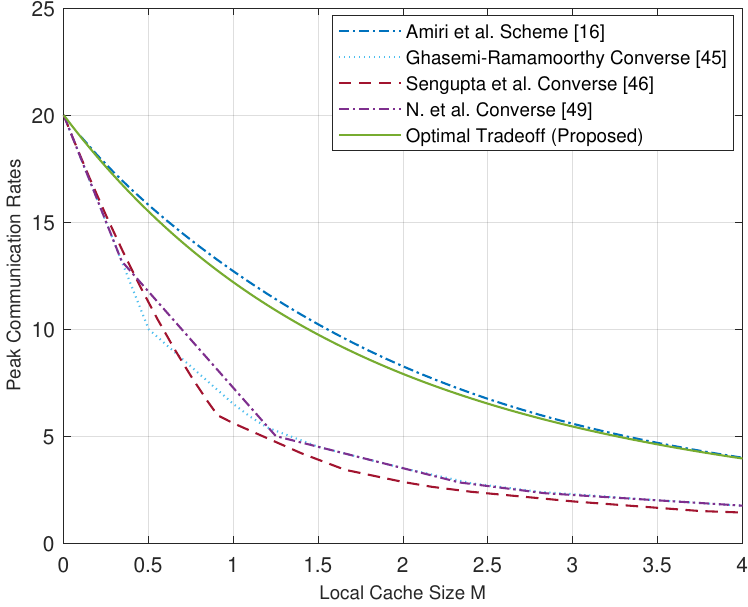}
  \caption{Peak rates for $N=20$, $K=40$. For this scenario, the best communication rate stated in prior works is achieved by the Amiri et al. scheme \cite{amiri2016coded}. The tightest prior converse bound in this scenario is provided by \cite{ghasemi15,sengupta15,prem2015critical}. }
\end{subfigure}
\vspace{3mm}
\caption{Numerical comparison between the optimal tradeoff and the state of the arts for the decentralized setting. Our results strictly improve the prior arts in both achievability and converse, for both average rate and peak rate.}
\label{fig:decent_compare}
\end{figure}

\begin{remark}
  We numerically compare our results with the state-of-the-art schemes and the converses for the decentralized setting. As shown in Fig. \ref{fig:decent_compare}, both the achievability scheme and the converse provided in our paper strictly improve the prior arts, for both average rate and peak rate.
\end{remark}

\section{Concluding Remarks}\label{sec:conclu}
		\begin{figure}[htbp]
    		\centering
    		\includegraphics[width=0.45\textwidth]{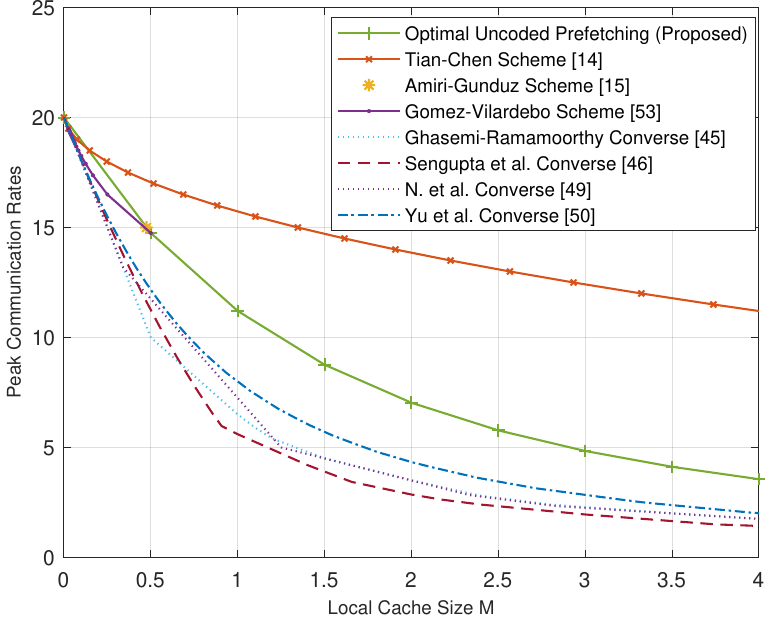}
    		\caption{Achievable peak communication rates for centralized schemes that allow coded prefetching. For $N=20$, $K=40$, we compare our proposed achievability scheme with prior-art coded-prefetching schemes \cite{tian2016caching,amiri2016fundamental}, prior-art converse bounds \cite{ghasemi15,sengupta15,prem2015critical}, and two recent results \cite{yu2017characterizing,gomez2016fundamental}. The achievability scheme proposed in this paper achieves the best performance to date in most cases, and is within a factor of $2$ optimal as shown in \cite{yu2017characterizing},  even compared with schemes that allow coded prefetching.
    }
    		\label{fig:coded}
    	\end{figure}

	In this paper, we characterized the rate-memory tradeoff for the coded caching problem with uncoded prefetching. To that end, we proposed the optimal caching schemes for both centralized setting and decentralized setting, and proved their exact optimality for both average rate and peak rate. 
	The techniques we introduced in this paper can
    be directly applied to many other problems, immediately improving their state of the arts. For instance, the achievability scheme proposed in this paper has already been applied in various different settings, achieving improved results   \cite{chugh2017improved,wan2017novel,yu2017characterizing}. Beyond these works, the techniques can also be  applied in directions such as online caching \cite{pedarsani13}, caching with non-uniform demands \cite{niesen13}, and hierarchical caching \cite{karamchandani14}, where improvements can be immediately achieved by directly plugging in our results.

	One interesting follow-up direction is to consider coded caching problem with coded placement. In this scenario, it has been shown that in the centralized setting, coded prefetching schemes can achieve better peak communication rates. For example, Figure \ref{fig:coded} shows that one can improve the peak communication rate by coded placement when the cache size is small. In a recent work \cite{yu2017characterizing} we have shown that, through a new converse bounding technique,  the achievability scheme we proposed in this paper is optimal within a factor of $2$. However, finding the exact optimal solution in this regime remains an open problem.

  \appendices
  \section{Proof of Lemma \ref{lemma:univ}}\label{app:keylemma}
  \begin{proof}
  
  The Proof of Lemma \ref{lemma:univ} is organized as follows: We start by proving a lower bound of the communication rate required for a single demand, i.e., $R_\epsilon^*(\boldsymbol{d},\boldsymbol{\mathcal{M}})$.  By averaging this lower bound over a single demand type $\mathcal{D}_{\boldsymbol{s}}$, we automatically obtain a lower bound for the rate $R_\epsilon^*(\boldsymbol{s},\boldsymbol{\mathcal{M}})$. Finally we bound the minimum possible $R_\epsilon^*(\boldsymbol{s},\boldsymbol{\mathcal{M}})$ over all prefetching schemes by solving for the minimum value of our derived lower bound.

  	We first use a genie-aided approach to derive a lower bound of $R_\epsilon^*(\boldsymbol{d},\boldsymbol{\mathcal{M}})$ for any demand $\boldsymbol{d}$ and for any prefetching $\boldsymbol{\mathcal{M}}$:	Given a demand $\boldsymbol{d}$, let $\mathcal{U}=\{u_1,...,u_{N_{\textup{e}}(\boldsymbol{d})}\}$ be an arbitrary subset with $N_{\textup{e}}(\boldsymbol{d})$ users that request distinct files.  We construct a virtual user whose cache is initially empty. Suppose for each $\ell\in\{1,...,N_{\textup{e}}(\boldsymbol{d})\}$, a genie fills the cache with the value of bits that are cached by $u_\ell$, but not from files requested by users in $\{u_1,...,u_{\ell-1}\}$. Then with all the cached information provided by the genie, the virtual user should be able to inductively decode all files requested by users in $\mathcal{U}$ upon receiving the message $X$.  Consequently, a lower bound on the communication rate $R_\epsilon^*(\boldsymbol{d},\boldsymbol{\mathcal{M}})$ can be obtained by applying a cut-set bound on the virtual user.

  	Specifically, we prove that the virtual user can decode all  $N_{\textup{e}}(\boldsymbol{d})$ requested files with high probability, by inductively decoding each file $d_{u_\ell}$ using the decoding function of user $u_\ell$, from $\ell=1$ to  $\ell={N_\textup{e}(\boldsymbol{d})}$:	Recall that any communication rate is $\epsilon$-achievable if the error probability of each decoding function is at most   $\epsilon$. Consequently, the probability that all $N_{\textup{e}}(\boldsymbol{d})$ decoding functions can correctly decode the requested files is at least $1-N_{\textup{e}}(\boldsymbol{d})\epsilon$. In this scenario, the virtual user can correctly decode all the files, given that at every single step of induction, all bits necessary for the decoding function have either been provided by the genie, or decoded in previous inductive steps. 
  	
  	Given this decodability, we can lower bound the needed communication load using Fano's inequality:
  	\begin{align}
  		R_\epsilon^*(\boldsymbol{d},&\boldsymbol{\mathcal{M}})F\geq \nonumber\\&H\left( \left.\{W_{d_{u_{\ell}}}\}_{\ell=1}^{N_{\textup{e}}(\boldsymbol{d})}\ \right|\ \textup{Bits cached by the virtual user} \right)\nonumber\\&-(1+ N^2_{\textup{e}}(\boldsymbol{d})\epsilon F ).
  	\end{align}
  	Recall that all bits in the library are i.i.d. and uniformly random, the cut-set bound in the above inequality essentially equals the number of bits in the $N_{\textup{e}}(\boldsymbol{d})$ requested files that are not cached by the virtual user. This set includes all bits in each file $d_{u_{\ell}}$ that are not cached by any users in $\{u_1,...,u_{\ell}\}$. Hence, the above lower bound is essentially
  	\begin{align}
  	&R_\epsilon^*(\boldsymbol{d},\boldsymbol{\mathcal{M}})F\geq\nonumber\\& \sum_{\ell=1}^{N_{\textup{e}}(\boldsymbol{d})} \sum_{j=1}^{F} \mathbbm{1}\left(B_{d_{u_\ell},j}\textup{ is not cached by any user in } \{u_1,...,u_{\ell}\}\right)\nonumber\\&-(1+ N^2_{\textup{e}}(\boldsymbol{d})\epsilon F ).
  	\end{align}
  	where $B_{i,j}$ denotes the $j$th bit in file $i$.
  	To simplify the discussion, we let $\mathcal{K}_{i,j}$ denote the subset of users that caches $B_{i,j}$. The above lower bound can be equivalently written as
  		\begin{align}
  	R_\epsilon^*(\boldsymbol{d},\boldsymbol{\mathcal{M}})F\geq &\sum_{\ell=1}^{N_{\textup{e}}(\boldsymbol{d})} \sum_{j=1}^{F} \mathbbm{1}\left(\mathcal{K}_{d_{u_\ell},j}\cap \{u_1,...,u_{\ell}\}=\emptyset\right)\nonumber\\&-(1+ N^2_{\textup{e}}(\boldsymbol{d})\epsilon F ).
  	\end{align}
  	
  	Using the above inequality, we derive a lower bound of the average rates as follows:
  	For any positive integer $i$, we denote the set of all permutations on $\{1,...,i\}$ by $\mathcal{P}_i$. Then, for each $p_1\in\mathcal{P}_K$ and $p_2\in\mathcal{P}_N$ given a demand $\boldsymbol{d}$, we define $\boldsymbol{d}(p_1,p_2)$ as a demand satisfying, for each user $k$, $d_k(p_1,p_2)=p_2 (d_{p_1^{-1}(k)})$. We can then apply the above bound to any demand $\boldsymbol{d}(p_1,p_2)$:
  		\begin{align}
  	R_\epsilon^*&(\boldsymbol{d}(p_1,p_2),\boldsymbol{\mathcal{M}})F\geq\nonumber\\& \sum_{\ell=1}^{N_{\textup{e}}(\boldsymbol{d})} \sum_{j=1}^{F} \mathbbm{1}\left(\mathcal{K}_{p_2(d_{u_\ell}),j}\cap \{p_1(u_1),...,p_1(u_{\ell})\}=\emptyset\right)\nonumber\\&-(1+ N^2_{\textup{e}}(\boldsymbol{d})\epsilon F ).\label{eq:genie_cut_p}
  	\end{align}
  	  	It is easy to verify that by taking the average of (\ref{eq:genie_cut_p}) over all pairs of $(p_1,p_2)$, only
  	the rates for demands in type $\mathcal{D}_{\boldsymbol{s}(\boldsymbol{d})}$ are counted, and each of them is counted the same number of times due to symmetry.
  	Consequently,  this approach provides us with a lower bound on the average rate within type $\mathcal{D}_{\boldsymbol{s}(\boldsymbol{d})}$, which is stated as follows:
  	\begin{align}
  	R^*_\epsilon({\boldsymbol{s}(\boldsymbol{d})},\boldsymbol{\mathcal{M}}) = & \frac{1}{K!N!} \sum_{p_1\in\mathcal{P}_K}\sum_{p_2\in\mathcal{P}_N} R^*_\epsilon(\boldsymbol{d}(p_1,p_2),\boldsymbol{\mathcal{M}}) \\\geq& \frac{1}{K!N!F} \sum_{p_1\in\mathcal{P}_K}\sum_{p_2\in\mathcal{P}_N }\sum_{\ell=1}^{N_{\textup{e}}(\boldsymbol{d})} \sum_{j=1}^{F} \nonumber\\ &\mathbbm{1}\left(\mathcal{K}_{p_2(d_{u_\ell}),j}\cap \{p_1(u_1),...,p_1(u_{\ell})\}=\emptyset\right)\nonumber\\&-\left(\frac{1}{F}+ N^2_{\textup{e}}(\boldsymbol{d})\epsilon  \right).\label{ineq:avebound1}
  	\end{align}
  	We aim to simplify the above lower bound, in order to find its minimum to prove Lemma \ref{lemma:univ}. To simplify this result, we first exchange the order of the summations and evaluate $\frac{1}{K!}\sum\limits_{p_1\in\mathcal{P}_K} \mathbbm{1}\left(\mathcal{K}_{p_2(d_{u_\ell}),j}\cap \{p_1(u_1),...,p_1(u_{\ell})\}=\emptyset\right)$. This is essentially the probability of selecting $\ell$ distinct users $\{p_1(u_1),...,p_1(u_{\ell})\}$ uniformly at random, such that none of them belongs to  $\mathcal{K}_{p_2(d_{u_\ell})}$. Out of the $\binom{K}{\ell}$ subsets, $\binom{K-|\mathcal{K}_{p_2(d_{u_\ell}),j}|}{\ell}$ of them satisfy this condition,\footnote{Recall that we define $\binom{n}{k}=0$ when $k>n$.} which gives the following identity:
  	\begin{align}
  	\frac{1}{K!}\sum\limits_{p_1\in\mathcal{P}_K} \mathbbm{1}\left(\mathcal{K}_{p_2(d_{u_\ell}),j}\cap \{p_1(u_1),...,p_1(u_{\ell})\}=\emptyset\right)=&\nonumber\\ \frac{\binom{K-|\mathcal{K}_{p_2(d_{u_\ell}),j}|}{\ell}}{\binom{K}{\ell}}&.\label{eq:ident1}
  	\end{align}
  	Hence, inequality (\ref{ineq:avebound1}) can be simplified based on (\ref{eq:ident1}) and the above discussion.
  		\begin{align}
  	R^*_\epsilon({\boldsymbol{s}(\boldsymbol{d})},\boldsymbol{\mathcal{M}}) \geq & \frac{1}{N!F} \sum_{p_2\in\mathcal{P}_N}\sum_{\ell=1}^{N_{\textup{e}}(\boldsymbol{d})} \sum_{j=1}^{F} \frac{1}{K!}\sum_{p_1\in\mathcal{P}_K} \nonumber\\&\mathbbm{1}\left(\mathcal{K}_{p_2(d_{u_\ell}),j}\cap \{p_1(u_1),...,p_1(u_{\ell})\}=\emptyset\right)\nonumber\\&-\left(\frac{1}{F}+ N^2_{\textup{e}}(\boldsymbol{d})\epsilon  \right)\\
  	=&\frac{1}{N!F} \sum_{p_2\in\mathcal{P}_N}\sum_{\ell=1}^{N_{\textup{e}}(\boldsymbol{d})} \sum_{j=1}^{F}  \frac{\binom{K-|\mathcal{K}_{p_2(d_{u_\ell}),j}|}{\ell}}{\binom{K}{\ell}}\nonumber\\&-\left(\frac{1}{F}+ N^2_{\textup{e}}(\boldsymbol{d})\epsilon  \right).\label{ineq:avebound2}
  	\end{align}
  	We further simplify this result by computing the summation over $p_2$ and $j$, and evaluating $\frac{1}{N!F}\sum\limits_{p_2\in\mathcal{P}_N}\sum\limits_{j=1}^{F} \binom{K-|\mathcal{K}_{p_2(d_{u_\ell}),j}|}{\ell}$. This is essentially the expectation of $\binom{K-|\mathcal{K}_{i,j}|}{\ell}$ over a uniformly randomly selected bit $B_{i,j}$.
  	Let $a_n$ denote the number of bits in the database that are cached by exactly $n$ users, then $|\mathcal{K}_{i,j}|=n$ holds for $\frac{a_n}{NF}$ fraction of the bits. Consequently, we have
  	\begin{align}
  	\frac{1}{N!F}\sum\limits_{p_2\in\mathcal{P}_N}\sum\limits_{j=1}^{F} \binom{K-|\mathcal{K}_{p_2(d_{u_\ell}),j}|}{\ell}=\sum_{n=0}^{K} \frac{a_n}{NF}\cdot \binom{K-n}{\ell}.
  	\end{align}
  	We simplify (\ref{ineq:avebound2}) using the above identity:
  		\begin{align}
  	R^*_\epsilon({\boldsymbol{s}(\boldsymbol{d})},\boldsymbol{\mathcal{M}}) \geq& \sum_{\ell=1}^{N_{\textup{e}}(\boldsymbol{d})} \frac{1}{N!F} \sum_{p_2\in\mathcal{P}_N}\sum_{j=1}^{F}  \frac{\binom{K-|\mathcal{K}_{p_2(d_{u_\ell}),j}|}{\ell}}{\binom{K}{\ell}}\nonumber\\&-\left(\frac{1}{F}+ N^2_{\textup{e}}(\boldsymbol{d})\epsilon  \right)\\
  	=&\sum_{\ell=1}^{N_{\textup{e}}(\boldsymbol{d})} \sum_{n=0}^{K} \frac{a_n}{NF}\cdot \frac{\binom{K-n}{\ell}}{\binom{K}{\ell}}-\left(\frac{1}{F}+ N^2_{\textup{e}}(\boldsymbol{d})\epsilon  \right)\label{ineq:avebound3}
  	\end{align}
  	It can be easily shown that
  	\begin{align}
  	\frac{\binom{K-n}{\ell}}{\binom{K}{\ell}}&=\frac{\binom{K-\ell}{n}}{\binom{K}{n}}
  	\end{align}
  	and
  	\begin{align}
  	\sum_{\ell=1}^{N_{\textup{e}}(\boldsymbol{d})}  \binom{K-\ell}{n}&=\binom{K}{n+1}-\binom{K-N_{\textup{e}}(\boldsymbol{d})}{n+1}.
  	\end{align}
  	Thus, we can rerwite (\ref{ineq:avebound3}) as
  		\begin{align}
  	R^*_\epsilon({\boldsymbol{s}(\boldsymbol{d})},\boldsymbol{\mathcal{M}}) \geq& \sum_{n=0}^{K} \frac{a_n}{NF}\cdot \frac{\binom{K}{n+1}-\binom{K-N_{\textup{e}}(\boldsymbol{d})}{n+1}}{\binom{K}{n}}\nonumber\\&-\left(\frac{1}{F}+ N^2_{\textup{e}}(\boldsymbol{d})\epsilon  \right).
  	\end{align}
  	Hence for any ${\boldsymbol{s}}\in\mathcal{S}$, by arbitrarily selecting a demand $\boldsymbol{d}\in \mathcal{D}_{\boldsymbol{s}}$ and applying the above inequality, the following bound holds for any prefetching $\boldsymbol{\mathcal{M}}$:
  	\begin{align}
  	R^*_\epsilon({\boldsymbol{s}},\boldsymbol{\mathcal{M}}) &\geq \sum_{n=0}^{K} \frac{a_n}{NF}\cdot \frac{\binom{K}{n+1}-\binom{K-N_{\textup{e}}(\boldsymbol{s})}{n+1}}{\binom{K}{n}}\nonumber\\&-\left(\frac{1}{F}+ N^2_{\textup{e}}(\boldsymbol{s})\epsilon  \right).
  	\end{align}
  	
  	After proving a lower bound of $R^*_\epsilon({\boldsymbol{s}},\boldsymbol{\mathcal{M}})$, we proceed to bound its minimum possible value over all prefetching schemes. 	Let $c_n$ denote the following sequence
  	\begin{align}
  	c_n &= \frac{\binom{K}{n+1}-\binom{K-N_{\textup{e}}(\boldsymbol{s})}{n+1}}{\binom{K}{n}}.
  	\end{align}
  	We have
  	\begin{align}
  	R^*_\epsilon({\boldsymbol{s}},\boldsymbol{\mathcal{M}}) &\geq \sum_{n=0}^{K} \frac{a_n}{NF}\cdot c_n-\left(\frac{1}{F}+ N^2_{\textup{e}}(\boldsymbol{s})\epsilon  \right).\label{ineq:ac}
  	\end{align}
  	We denote the lower convex envelope of $c_n$, i.e., the lower convex envelope of points $\{(t,c_t)\ |\ t\in\{0,1,...,K\} \}$, by $\textup{Conv}(c_t)$.
  	Note that $c_n$ is a decreasing sequence, so its lower convex envelope is  a decreasing and convex function.   
  	
  	Because the following holds for every prefetching:
  	\begin{align}
  	\sum_{n=0}^K a_n&=NF,\\
  	\sum_{n=0}^K n a_n&\leq NFt,
  	\end{align}
  	we can lower bound (\ref{ineq:ac}) using Jensen's inequality and the monotonicity of $\textup{Conv}(c_t)$: 
  	\begin{align}
  	R^*_\epsilon({\boldsymbol{s}},\boldsymbol{\mathcal{M}}) &\geq \textup{Conv}(c_t)-\left(\frac{1}{F}+ N^2_{\textup{e}}(\boldsymbol{s})\epsilon  \right).
  	\end{align}
  	Consequently, 
  	\begin{align}
  	\min_{\boldsymbol{\mathcal{M}}}R^*_\epsilon({\boldsymbol{s}},\boldsymbol{\mathcal{M}}) \geq& \min_{\boldsymbol{\mathcal{M}}} \textup{Conv}(c_t)-\left(\frac{1}{F}+ N^2_{\textup{e}}(\boldsymbol{s})\epsilon  \right)\\
  	= &\textup{Conv}(c_t)-\left(\frac{1}{F}+ N^2_{\textup{e}}(\boldsymbol{s})\epsilon  \right)\\
  	=&\textup{Conv}\left(\frac{\binom{K}{t+1}-\binom{K-N_{\textup{e}}(\boldsymbol{s})}{t+1}}{\binom{K}{t}}\right)\nonumber\\&-\left(\frac{1}{F}+ N^2_{\textup{e}}(\boldsymbol{s})\epsilon  \right).
  	\end{align}

  \end{proof}
  \section{Minimum Peak Rate for Centralized Caching}\label{app:worst}

 Consider a caching problem with $K$ users, a database of $N$ files, and a local cache size of $M$ files for each user.  We define the rate-memory tradeoff for the peak rate as follows:
  Similar to the average rate case, 
  for each prefetching $\boldsymbol{\mathcal{M}}$, let $R_{\epsilon,\textup{peak}}^*(\boldsymbol{\mathcal{M}})$ denote the peak rate, defined as 
  \begin{equation}
  R_{\epsilon,\textup{peak}}^*(\boldsymbol{\mathcal{M}})=\max_{\boldsymbol{d}} R_\epsilon^*(\boldsymbol{d},\boldsymbol{\mathcal{M}}).\nonumber
  \end{equation}	
  We aim to find the minimum peak rate $R^*_{\textup{peak}}$, where 
  \begin{equation}
  R_{\textup{peak}}^*=\sup_{\epsilon>0}\adjustlimits \limsup_{F\rightarrow+\infty}\min_{\boldsymbol{\mathcal{M}}} R_{\epsilon,\textup{peak}}^*(\boldsymbol{\mathcal{M}}),\nonumber
  \end{equation}
  which is a function of $N$, $K$, and $M$.

  Now we prove Corollary \ref{cor:worst}, which completely characterizes the value of $R_{\textup{peak}}^*$.

  	\begin{proof}
  		It is easy to show that the rate stated in Corollary \ref{cor:worst} can be exactly achieved using the caching scheme introduced in Section \ref{sec:opt}. Hence, we focus on proving the optimality of the proposed coding scheme. 

  		Recall the definitions of statistics and types (see section \ref{sec:conv}). Given a prefetching $\boldsymbol{\mathcal{M}}$ and statistics $\boldsymbol{s}$, we define the peak rate within type $\mathcal{D}_{\boldsymbol{s}}$, denoted by $R_{\epsilon,\textup{peak}}^*(\boldsymbol{s},\boldsymbol{\mathcal{M}})$, as 
  		\begin{equation}
  		R_{\epsilon,\textup{peak}}^*(\boldsymbol{s},\boldsymbol{\mathcal{M}})=\max_{\boldsymbol{d}\in \mathcal{D}_{\boldsymbol{s}}} R_{\epsilon}^*(\boldsymbol{d},\boldsymbol{\mathcal{M}}).
  		\end{equation}	
  		
  		Note that 
  		\begin{align}
  		R^*_{\textup{peak}}&=\sup_{\epsilon>0}\adjustlimits \limsup_{F\rightarrow+\infty}\min_{{\boldsymbol{\mathcal{M}}}} \max_{\boldsymbol{s}} R^*_{\epsilon,\textup{peak}}({\boldsymbol{s}}, \boldsymbol{\mathcal{M}}) \\
  		&\geq \sup_{\epsilon>0}\adjustlimits \limsup_{F\rightarrow+\infty}\max_{\boldsymbol{s}}\min_{{\boldsymbol{\mathcal{M}}}} R^*_{\epsilon,\textup{peak}}({\boldsymbol{s}}, \boldsymbol{\mathcal{M}}).\label{ineq:54}
  		\end{align}
  		Hence, in order to lower bound $R^*$, 
  		it is sufficient to bound the minimum value of $R^*_{\epsilon,\textup{peak}}({\boldsymbol{s}}, \boldsymbol{\mathcal{M}})$ for each type $\mathcal{D}_{\boldsymbol{s}}$ individually. 
  		Using Lemma \ref{lemma:univ}, the following bound holds for each $s\in\mathcal{S}$:
  		\begin{align}
  	\min_{{\boldsymbol{\mathcal{M}}}}	R^*_{\epsilon,\textup{peak}}({\boldsymbol{s}}, \boldsymbol{\mathcal{M}})\geq&
  	\min_{{\boldsymbol{\mathcal{M}}}}	R_{\epsilon}^*({\boldsymbol{s}}, \boldsymbol{\mathcal{M}})\\\geq& \textup{Conv}\left( \frac{\binom{K}{t+1}-\binom{K-N_{\textup{e}}(\boldsymbol{s})}{t+1}}{\binom{K}{t}}\right)\nonumber\\&-\left(\frac{1}{F}+ N^2_{\textup{e}}(\boldsymbol{s})\epsilon  \right)\label{ineq:typ-worst}.
  		\end{align}
  		
  		Consequently, 
  		\begin{align}
  		R^*_{\textup{peak}}\geq& \sup_{\epsilon>0}\adjustlimits \limsup_{F\rightarrow+\infty}
  		 \max_{\boldsymbol{s}}\textup{Conv}\left( \frac{\binom{K}{t+1}-\binom{K-N_{\textup{e}}(\boldsymbol{s})}{t+1}}{\binom{K}{t}}\right)\nonumber\\&-\left(\frac{1}{F}+ N^2_{\textup{e}}(\boldsymbol{s})\epsilon  \right)\\
  		 =&\textup{Conv}\left( \frac{\binom{K}{t+1}-\binom{K-\min\{N,K\}}{t+1}}{\binom{K}{t}}\right).
  		\end{align}  			
  	\end{proof}
  	
  	 \begin{remark}[Universal Optimality of Symmetric Batch Prefetching - Peak Rate] \label{remark:peakuniv}
  	 	Inequality (\ref{ineq:typ-worst}) characterizes the minimum peak rate given a type $\mathcal{D}_{\boldsymbol{s}}$, if the prefetching $\boldsymbol{\mathcal{M}}$ can be designed based on $\boldsymbol{s}$. However,
  	 	for (\ref{ineq:54}) to be tight, the peak rate for each different type has to be minimized on the same prefetching. Surprisingly, such an optimal prefetching exists, an example being the symmetric batch prefetching, according to Section \ref{sec:opt}. This indicates that the symmetric batch prefetching is also universally optimal for all types in terms of peak rates.
  	 \end{remark}
  	
  \section{Proof of Theorem \ref{thm:ave-dec}}\label{app:decent-ave}
  
  To completely characterize $\mathcal{R}$, we propose decentralized caching schemes to achieve all points in $\mathcal{R}$. We also prove a matching information-theoretic outer
  		bound of the achievable regions, which implies that none of the points outside $\mathcal{R}$ are achievable.

\subsection{The Optimal Decentralized Caching Scheme}\label{app:decent-opt}

To prove the achievability of $\mathcal{R}$, we need to provide an optimal decentralized prefetching scheme ${P}_{\mathcal{M};F}$, an optimal delivery scheme for every possible user demand $\boldsymbol{d}$ that achieves the corner point in $\mathcal{R}$, and a valid decoding algorithm for the users. The main idea of our proposed achievability scheme is to first design a decentralized prefetching scheme, such that we can view the resulting content delivery problem as a list of sub-problems that can be individually solved using the techniques we already developed for the centralized setting. Then we optimally solve this delivery problem by greedily applying our proposed centralized delivery and decoding scheme.

We consider the following optimal prefetching scheme: all users cache $\frac{MF}{N}$ bits in each file uniformly and independently. This prefetching scheme was originally proposed in \cite{maddah-ali13}. For convenience, we refer to this prefetching scheme as \textit{uniformly random prefetching scheme}.
Given this prefetching scheme, each bit in the database is cached by a random subset of the $K$ users. 

During the delivery phase, we first greedily categorize all the bits based on the number of users that cache the bit, then within each category, we deliver the corresponding messages in an opportunistic way using the delivery scheme described in Section \ref{sec:opt} for centralized caching.
 For any demand $\boldsymbol{d}$ where $K$ users are making requests, and any realization of the prefetching on these $K$ users, we divide the bits in the database into $K+1$ sets:
 For each $j\in\{0,1,...,K\}$, let $\mathcal{B}_j$ denote the bits that are cached by exactly $j$ users. To deliver the requested files to the $K$ users, it is sufficient to deliver all the corresponding bits in each $\mathcal{B}_j$ individually.

 Within each $\mathcal{B}_j$, first note that with high probability for large $F$, the number of bits that belong to each file is approximately  $\binom{K}{j}(\frac{M}{N})^j(1-\frac{M}{N})^{K-j}F+o(F)$, which is the same across all files. Furthermore,  for any subset $\mathcal{K}\subseteq\{1,...,K\}$ of size $j$, a total of $(\frac{M}{N})^j(1-\frac{M}{N})^{K-j}F+o(F)$ bits in file $i$ are exclusively cached by users in $\mathcal{K}$, which is $1/\binom{K}{j}$ fraction of the bits in $\mathcal{B}_j$ that belong to file $i$. This is effectively the symmetric batch prefetching, and hence we can directly apply the same delivery and decoding scheme to deliver all the requested bits within this subset. 
 
 Recall that in the centralized setting, when each file has a size $F$ and each bit is cached by exactly $t$ users, our proposed delivery scheme achieves a communication load of $\frac{\binom{K}{t+1}-\binom{K-N_{\textup{e}}(\boldsymbol{d})}{t+1}}{\binom{K}{t}} F$. Then to deliver all requested bits within $\mathcal{B}_j$, where the equivalent file size approximately equals $\binom{K}{j}(\frac{M}{N})^j(1-\frac{M}{N})^{K-j}F$, we need a communication rate of $(\frac{M}{N})^j(1-\frac{M}{N})^{K-j} \left(\binom{K}{j+1}-\binom{K-N_{\textup{e}}(\boldsymbol{d})}{j+1}\right)$.
 
 Consequently, by applying the delivery scheme for all $j\in\{0,1,...,K\}$, we achieve a total communication rate of 
 \begin{align}
 R_K=&\sum_{j=0}^{K}\left(\frac{M}{N}\right)^j\left(1-\frac{M}{N}\right)^{K-j} \nonumber\\&\ \ \ \ \ \cdot\left(\binom{K}{j+1}-\binom{K-N_{\textup{e}}(\boldsymbol{d})}{j+1}\right)\\
 =&\frac{N-M}{M}\left(1-\left(1-\frac{M}{N}\right)^{N_{\textup{e}}(\boldsymbol{d})}\right)
 \end{align}
 for any demand $\boldsymbol{d}$.
 Hence, for each $K$ we achieve an average rate of $\mathbb{E}[\frac{N-M}{M}(1-\left(1-\frac{M}{N}\right)^{N_{\textup{e}}(\boldsymbol{d})})]$, which dominates all points in $\mathcal{R}$. This provides a tight inner bound for Theorem \ref{thm:ave-dec}.
  
\subsection{Converse}
    To prove an outer bound of $\mathcal{R}$, i.e., bounding all possible rate vectors $\{R_K\}_{K\in\mathbb{N}}$ that can be achieved by a prefetching scheme, it is sufficient to bound each entry of the vector individually, by providing a lower bound of $R^*_K(P_{{\mathcal{M}};F})$ that holds for all prefetching schemes. To obtain such a lower bound, for each $K\in\mathbb{N}$ we divide the set of all possible demands into types, and  derive the minimum average rate within each type separately.      	  	
    
   	For any statistics $\boldsymbol{s}$, we let $R^*_{\epsilon,K}(\boldsymbol{s}, P_{{\mathcal{M}}})$ denote the average rate within type $\mathcal{D}_{\boldsymbol{s}}$. Rigorously, 
   	\begin{align}
   R^*_{\epsilon,K}(\boldsymbol{s}, P_{{\mathcal{M}}})&=\frac{1}{|\mathcal{D}_{\boldsymbol{s}}|}\sum_{\boldsymbol{d}\in \mathcal{D}_{\boldsymbol{s}}}R^*_{\epsilon, K}(\boldsymbol{d},P_{{\mathcal{M}}}).
   	\end{align}
   	The minimum value of $R^*_{\epsilon,K}(\boldsymbol{s}, P_{{\mathcal{M}}})$ is lower bounded by the following lemma:
   \begin{lemma}\label{lemma:univ-dec}
   	Consider a decentralized caching problem with $N$ files and a local cache size of $M$ files for each user. For any type $\mathcal{D}_{\boldsymbol{s}}$, where $K$ users are making requests, the minimum value of $R^*_{\epsilon, K}(\boldsymbol{s}, P_{{\mathcal{M}}})$ is lower bounded by
   	\begin{align}
   	\min_{P_{{\mathcal{M}}}}	
   	R^*_{\epsilon,K}(\boldsymbol{s}, P_{{\mathcal{M}}})\geq& \frac{M-N}{M}\left(1-\left(1-\frac{M}{N}\right)^{N_{\textup{e}}(s)}\right)\nonumber\\&-\left(\frac{1}{F}+N_{\textup{e}}^2 (\boldsymbol{s})\epsilon\right).
   	\end{align}
   \end{lemma}
   
      		 \begin{remark}
      		 	As proved in Appendix \ref{app:decent-opt}, the rate $R^*_{\epsilon,K}(\boldsymbol{s}, P_{{\mathcal{M}}})$ for any statistics $\boldsymbol{s}$ and any $K$ can be simultaneously minimized using the uniformly random prefetching scheme. This demonstrates that the uniformly random prefetching scheme is universally optimal for the decentralized caching problem in terms of average rates. 
      		 \end{remark}
   \begin{proof}
   	To prove Lemma \ref{lemma:univ-dec}, we first consider a class of \emph{generalized demands}, where not all users in the caching systems are required to request a file. We define generalized demand $\boldsymbol{d}=(d_1,...,d_K)\in \{0,1,...,N\}^K$, where
   a nonzero $d_k$ denotes the index of the file requested by $k$, while $d_k=0$ indicates that user $k$ is not making a request. We define statistics and their corresponding types in the same way, and let $R^*_{\epsilon,K}(\boldsymbol{s} ,\boldsymbol{\mathcal{M}})$ denote the centralized average rate on a generalized type $\mathcal{D}_{\boldsymbol{s}}$ given prefetching $\boldsymbol{\mathcal{M}}$. 
   	
   	For a centralized caching problem, we can easily generalize Lemma \ref{lemma:univ} to the following lemma for the generalized demands: 
   		 \begin{lemma}\label{lemma:univ-gen}
   		 	Consider a caching problem with $N$ files, $K$ users, and a local cache size of $M$ files for each user. For any generalized type $\mathcal{D}_{\boldsymbol{s}}$, the minimum value of $R^*_{\epsilon,K}(\boldsymbol{s}, \boldsymbol{\mathcal{M}})$ is lower bounded by
   		 	\begin{align}
   		 	\min_{\boldsymbol{\mathcal{M}}}	
   		 	R^*_{\epsilon,K}(\boldsymbol{s} ,\boldsymbol{\mathcal{M}})\geq& \textup{Conv}\left( \frac{\binom{K}{t+1}-\binom{K-N_{\textup{e}}(\boldsymbol{s})}{t+1}}{\binom{K}{t}}\right)\nonumber\\&-\left(\frac{1}{F}+N_{\textup{e}}^2 (\boldsymbol{s})\epsilon\right),
   		 	\end{align}
   		 	where $\textup{Conv}(f(t))$ denotes the lower convex envelope of the following points: $\{(t,f(t))\ | \ t\in\{0,1,...,K\}\}$.
   		 \end{lemma}
   	  The above lemma can be proved exactly the same way as we proved Lemma \ref{lemma:univ}, and the universal optimality of symmetric batch prefetching still holds for the generalized demands.

   	  For a decentralized caching problem, we can also generalize the definition of $R^*_{\epsilon,K}(\boldsymbol{s}, P_{{\mathcal{M}}})$ correspondingly. 
   	  We can easily prove that, when a decentralized caching scheme is used, the expected value of $R^*_{\epsilon,K}(\boldsymbol{s}, \boldsymbol{\mathcal{M}})$ is no greater than $R^*_{\epsilon,K}(\boldsymbol{s}, P_{{\mathcal{M}}})$. Consequently, 
   	  \begin{align}
   	  R^*_{\epsilon,K}(\boldsymbol{s}, P_{{\mathcal{M}}})\geq& \mathbb{E}_{\boldsymbol{\mathcal{M}}}[R^*_{\epsilon,K}(\boldsymbol{s}, \boldsymbol{\mathcal{M}})]\\
   	  \geq& \textup{Conv}\left( \frac{\binom{K}{t+1}-\binom{K-N_{\textup{e}}(\boldsymbol{s})}{t+1}}{\binom{K}{t}}\right)\nonumber\\&-\left(\frac{1}{F}+N_{\textup{e}}^2 (\boldsymbol{s})\epsilon\right),
   	  \end{align}
   	  for any generalized type $\mathcal{D}_{\boldsymbol{s}}$ and for any $P_{{\mathcal{M}}}$.
   	  
   	  Now we prove that value $R^*_{\epsilon,K}(\boldsymbol{s}, P_{{\mathcal{M}}})$ is independent of parameter $K$ given $\boldsymbol{s}$ and $ P_{{\mathcal{M}}}$:
   	  Consider a generalized statistic $\boldsymbol{s}$. Let $K_{\boldsymbol{s}}=\sum\limits_{i=1}^{N} s_i$, which equals the number of active users for demands in $\mathcal{D}_{\boldsymbol{s}}$. For any caching system with $K>K_{\boldsymbol{s}}$ users, and for any subset $\mathcal{K}$ of $K_{\boldsymbol{s}}$ users, let $\mathcal{D}_{\mathcal{K}}$ denote the set of demands in $\mathcal{D}_{\boldsymbol{s}}$ where only users in $\mathcal{K}$ are making requests.
   	  Note that $\mathcal{D}_{\boldsymbol{s}}$ equals the union of disjoint sets  $\mathcal{D}_{\mathcal{K}}$ for all subsets $\mathcal{K}$ of size $K_{\boldsymbol{s}}$. Thus we have,
   	  \begin{align}
   	  R^*_{\epsilon,K}(\boldsymbol{s}, P_{{\mathcal{M}}})&=\frac{1}{|\mathcal{D}_{\boldsymbol{s}}|} \sum_{\boldsymbol{d}\in\mathcal{D}_{\boldsymbol{s}}} R^*_{\epsilon,K}(\boldsymbol{d}, P_{{\mathcal{M}}}) \\
   	  &=\frac{1}{|\mathcal{D}_{\boldsymbol{s}}|} \sum_{\mathcal{K}:|\mathcal{K}|=K_{\boldsymbol{s}}} \sum_{\boldsymbol{d}\in\mathcal{D}_{\mathcal{K}}} R^*_{\epsilon,K}(\boldsymbol{d}, P_{{\mathcal{M}}})\\
   	  &=\frac{1}{|\mathcal{D}_{\boldsymbol{s}}|} \sum_{\mathcal{K}:|\mathcal{K}|=K_{\boldsymbol{s}}}
   	  |\mathcal{D}_{\mathcal{K}}| R^*_{\epsilon,K_{\boldsymbol{s}}}(\boldsymbol{s}, P_{{\mathcal{M}}})\\
   	  &=R^*_{\epsilon,K_{\boldsymbol{s}}}(\boldsymbol{s}, P_{{\mathcal{M}}}).
   	  \end{align}
   	  
   	  Consequently,
   	  \begin{align}
   	  R^*_{\epsilon,K_{\boldsymbol{s}}}(\boldsymbol{s}, P_{{\mathcal{M}}})=&\lim_{K\rightarrow+\infty} R^*_{\epsilon,K}(\boldsymbol{s}, P_{{\mathcal{M}}})\\
   	  \geq& \lim_{K\rightarrow+\infty} \textup{Conv}\left( \frac{\binom{K}{t+1}-\binom{K-N_{\textup{e}}(\boldsymbol{s})}{t+1}}{\binom{K}{t}}\right)\nonumber\\&-\left(\frac{1}{F}+N_{\textup{e}}^2 (\boldsymbol{s})\epsilon\right)\\
   	  =&\frac{M-N}{M}\left(1-\left(1-\frac{M}{N}\right)^{N_{\textup{e}}(s)}\right)\nonumber\\&-\left(\frac{1}{F}+N_{\textup{e}}^2 (\boldsymbol{s})\epsilon\right).
   	  \end{align}   	  
   	Because the above lower bound is independent of the prfetching distribution $P_{\mathcal{\boldsymbol{M}}}$, the minimum value of  $R^*_{\epsilon,K_{\boldsymbol{s}}}(\boldsymbol{s}, P_{{\mathcal{M}}})$ over all possible prefetchings is also bounded by the same formula. This completes the proof of Lemma \ref{lemma:univ-dec}.
   \end{proof}
   
   From Lemma \ref{lemma:univ-dec}, the following bound holds by definition
   \begin{align}
   R^*_K(P_{{\mathcal{M}};F})&=\sup_{\epsilon>0} \limsup_{F'\rightarrow+\infty}\mathbb{E}_{\boldsymbol{s}} [R^*_{\epsilon,K}(\boldsymbol{s}, P_{{\mathcal{M}};F}(F=F'))]\\&\geq \mathbb{E}_{\boldsymbol{d}} \left[\frac{M-N}{M}\left(1-\left(1-\frac{M}{N}\right)^{N_{\textup{e}}(\boldsymbol{d})}\right)\right]
   \end{align}
   for any $K\in \mathbb{N}$ and for any prefetching scheme $P_{{\mathcal{M}};F}$. Consequently, any vector $\{R_K\}_{K\in\mathbb{N}}$ in $\mathcal{R}$ satisfies
   \begin{align}
   R_K&\geq \min_{P_{{\mathcal{M}};F}}R^*_K(P_{{\mathcal{M}};F})\\&\geq 
   \mathbb{E}_{\boldsymbol{d}} \left[\frac{M-N}{M}\left(1-\left(1-\frac{M}{N}\right)^{N_{\textup{e}}(\boldsymbol{d})}\right)\right],
   \end{align}
   for any $K\in \mathbb{N}$.
   Hence,
 \begin{align}
			\mathcal{R}\subseteq&\left\{\{R_K\}_{K\in\mathbb{N}}\ \vphantom{\left.R_K\geq \mathbb{E}_{\boldsymbol{d}}\left[\frac{N-M}{M} \left(1-\left(\frac{N-M}{N}\right)^{N_{\textup{e}}(\boldsymbol{d})}\right)\right]\right\}} \right|\nonumber\\
			& \ \  \left.R_K\geq \mathbb{E}_{\boldsymbol{d}}\left[\frac{N-M}{M} \left(1-\left(\frac{N-M}{N}\right)^{N_{\textup{e}}(\boldsymbol{d})}\right)\right]\right\}.
			\end{align}

  \section{Proof of Corollary \ref{cor:dec}}\label{app:decent-worst}
  \begin{proof}
  It is easy to show that all points in $\mathcal{R}_{\textup{peak}}$ can be achieved using the decentralized caching scheme introduced in Appendix \ref{app:decent-opt}. Hence, we focus on proving the optimality of the proposed decentralized caching scheme. 
   Similar to the average rate case, we prove an outer bound of $\mathcal{R}_{\textup{peak}}$ by bounding  $R^*_{K,\textup{peak}}(P_{{\mathcal{M}};F})$ for each $K\in\mathbb{N}$ individually. To do so, we divide the set of all possible demands into types, and derive the minimum average rate within each type separately.

  	Recall the definitions of statistics and types (see section \ref{sec:conv}). Given a caching system with $N$ files, $K$ users, a prefetching distribution $P_{{\mathcal{M}}}$, and a statistic $\boldsymbol{s}$, we define the peak rate within type $\mathcal{D}_{\boldsymbol{s}}$, denoted by $R_{\epsilon, K,\textup{peak}}^*(\boldsymbol{s},P_{{\mathcal{M}}})$, as 
  	\begin{equation}
  	R_{\epsilon,K,\textup{peak}}^*(\boldsymbol{s},P_{{\mathcal{M}}})=\max_{\boldsymbol{d}\in \mathcal{D}_{\boldsymbol{s}}} R^*_{\epsilon,K}(\boldsymbol{d},P_{{\mathcal{M}}}).
  	\end{equation}	
  	
  	Note that any point $\{R_K\}_{K\in\mathbb{N}}$ in $\mathcal{R}_{\textup{peak}}$ satisfies 
  	\begin{align}
  	R_K&\geq \inf_{P_{{\mathcal{M}};F}} R^*_{K,\textup{peak}}(P_{{\mathcal{M}};F})\\&= \inf_{P_{{\mathcal{M}};F}} \sup_{\epsilon>0}\adjustlimits \limsup_{F'\rightarrow+\infty}\max_{\boldsymbol{s}\in \mathcal{D}} [R^*_{\epsilon,K,{\textup{peak}}}(\boldsymbol{s}, P_{{\mathcal{M};F}}(F=F'))]
  	\end{align}
  	for any $K\in \mathbb{N}$. 
  	We have the following from min-max inequality
  	\begin{align}
  	    	R_K&\geq \sup_{\epsilon>0}\adjustlimits \limsup_{F\rightarrow+\infty}\max_{\boldsymbol{s}\in \mathcal{D}} [\min_{P_{{\mathcal{M}}}}  R^*_{\epsilon,K,{\textup{peak}}}(\boldsymbol{s}, P_{{\mathcal{M}}})].\label{ineq:78}
  	\end{align}
  	Hence, in order to outer bound $\mathcal{R}_{\textup{peak}}$, 
  	it is sufficient to bound the minimum value of $R_{\epsilon,K,\textup{peak}}^*(\boldsymbol{s},P_{{\mathcal{M}}})$ for each type $\mathcal{D}_{\boldsymbol{s}}$ individually. 
  	
  	Using Lemma \ref{lemma:univ-dec}, the following bound holds for each $s\in\mathcal{S}$:
  	\begin{align}
  	\min_{P_{{\mathcal{M}}}} R_{\epsilon,K,\textup{peak}}^*(\boldsymbol{s},P_{{\mathcal{M}}})\geq&
  		\min_{P_{{\mathcal{M}}}} R_{\epsilon,K}^*(\boldsymbol{s},P_{{\mathcal{M}}})\\\geq& \frac{M-N}{M}\left(1-\left(1-\frac{M}{N}\right)^{N_{\textup{e}}(s)}\right) \nonumber\\&-\left(\frac{1}{F}+N_{\textup{e}}^2 (\boldsymbol{s})\epsilon\right) \label{ineq:typ-dec-worst}.
  	\end{align}
  	Hence for any $\{R_K\}_{K\in\mathbb{N}}$,
  	\begin{align}
  	R_K\geq& \sup_{\epsilon>0}\adjustlimits \limsup_{F\rightarrow+\infty}\max_{\boldsymbol{s}}\left[ \frac{M-N}{M}\left(1-\left(1-\frac{M}{N}\right)^{N_{\textup{e}}(s)}\right) \right.\nonumber\\&-\left.\left(\frac{1}{F}+N_{\textup{e}}^2 (\boldsymbol{s})\epsilon\right) \vphantom{\left(1-\frac{M}{N}\right)^{N_{\textup{e}}(s)}}\right]\\=&\frac{M-N}{M}\left(1-\left(1-\frac{M}{N}\right)^{\min\{N,K\}}\right). 
  	\end{align}
  	Consequently, 
  	\begin{align}
  	\mathcal{R}_{\textup{peak}}\subseteq&\left\{\vphantom{\left(1-\frac{M}{N}\right)^{N_{\textup{e}}(s)}}\{R_K\}_{K\in\mathbb{N}}\ \right|\nonumber\\
  	&\ \  \left. R_K\geq \frac{N-M}{M} \left(1-\left(\frac{N-M}{N}\right)^{\min\{N,K\}}\right)\right\}.
  	\end{align}  			
  \end{proof}
  
  		 \begin{remark}
      		 	According to the above discussion, the rate $R^*_{\epsilon,K,\textup{peak}}(\boldsymbol{s}, P_{{\mathcal{M}}})$ for any statistics $\boldsymbol{s}$ and any $K$ can be simultaneously minimized using the uniformly random prefetching scheme. This indicates that the uniformly random prefetching scheme is universally optimal for all types in terms of peak rates.
  \end{remark}

\bibliographystyle{ieeetr}
\bibliography{uache_checked}

\section*{Biographies}

\begin{IEEEbiographynophoto}{Qian Yu} (S'16) is pursuing his Ph.D. degree in Electrical Engineering at University of Southern California (USC), Viterbi School of Engineering. He received his M.Eng. degree in Electrical Engineering and B.S. degree in EECS and Physics, both from Massachusetts Institute of Technology (MIT). His interests span information theory, distributed computing, and many other problems math-related.

Qian received the Jack Keil Wolf ISIT Student Paper Award in 2017.
He is also a Qualcomm Innovation Fellowship finalist in 2017, and received the Annenberg Graduate Fellowship in 2015. 
\end{IEEEbiographynophoto}

\begin{IEEEbiographynophoto}{Mohammad Ali Maddah-Ali} (S'03-M'08) received the B.Sc. degree from Isfahan University of Technology, and the M.A.Sc. degree from the University of Tehran, both in electrical engineering. From 2002 to 2007, he was with the Coding and Signal Transmission Laboratory (CST Lab), Department of Electrical and Computer Engineering, University of Waterloo, Canada, working toward the Ph.D. degree. From 2007 to 2008, he worked at the Wireless Technology Laboratories, Nortel Networks, Ottawa, ON, Canada.  From 2008 to 2010, he was a post-doctoral fellow in the Department of Electrical Engineering and Computer Sciences at the University of California at Berkeley. Then, he joined Bell Labs, Holmdel, NJ, as a communication research scientist. Recently, he started working at Sharif University of Technology, as a faculty member.
 
Dr. Maddah-Ali is a recipient of NSERC Postdoctoral Fellowship in 2007, a best paper award from IEEE International Conference on Communications (ICC) in 2014, the IEEE Communications Society and IEEE Information Theory Society Joint Paper Award in 2015, and the IEEE Information Theory Society Joint Paper Award in 2016. 
\end{IEEEbiographynophoto}

\begin{IEEEbiographynophoto}{A. Salman Avestimehr} (S'03-M'08-SM'17) is an Associate Professor at the Electrical Engineering Department of University of Southern California. He received his Ph.D. in 2008 and M.S. degree in 2005 in Electrical Engineering and Computer Science, both from the University of California, Berkeley. Prior to that, he obtained his B.S. in Electrical Engineering from Sharif University of Technology in 2003. His research interests include information theory, the theory of communications, and their applications to distributed computing and data analytics.

Dr. Avestimehr has received a number of awards, including the Communications Society and Information Theory Society Joint Paper Award, the Presidential Early Career Award for Scientists and Engineers (PECASE) for ``pushing the frontiers of information theory through its extension to complex wireless information networks'', the Young Investigator Program (YIP) award from the U. S. Air Force Office of Scientific Research, the National Science Foundation CAREER award, and the David J. Sakrison Memorial Prize. He is currently an Associate Editor for the IEEE Transactions on Information Theory.
\end{IEEEbiographynophoto}

\end{document}